\documentclass[11pt]{article}
\usepackage{fullpage}
 \usepackage{hyperref}
 \hypersetup{
    pdftitle = {Secure Type CSF 17}
}
\usepackage[utf8]{inputenc}
\usepackage{amsmath}
\usepackage{amssymb}
\usepackage{commands}
\usepackage{types-commands}
\usepackage{mathpartir}
\usepackage{multicol}
\usepackage{multirow}
\usepackage{varwidth}
\usepackage{thm-restate}
\usepackage{color}
\usepackage{extarrows} 
\usepackage{etoolbox}

\newif\ifcomments
\commentstrue

\newtheorem{definition}{Definition}

\newtheorem{lemma}{Lemma}

\newtheorem{example}{Example}
\newenvironment{proof}{\textit{Proof.}}

\newtoggle{techreport}
\toggletrue{techreport}

\newcommand{\trlinebreak}{\linebreak}

\begin{document}

\title{Types for Location and Data Security \\ in Cloud Environments}

\author{
        {\rm Ivan Gazeau\footnote{LORIA, INRIA Nancy -
     Grand-Est,France}}          
        \and {\rm Tom Chothia\footnote{School of Computer Science, Univ. of Birmingham, UK}}
        \and {\rm Dominic Duggan\footnote{Department of Computer Science,
        Stevens Institute of Technology, USA.}}
}

\date{Stevens Technical Report CS-2017-1}

\maketitle

\begin{abstract}
Cloud service providers are often trusted to be genuine, the damage caused by being discovered to be attacking their own customers outweighs any benefits such attacks could reap. On the other hand, it is expected that some cloud service users may be actively malicious. In such an open system, each location may run code which has been developed independently of other locations (and which may be secret). In this paper, we present a typed language which ensures that the access restrictions put on data on a particular device will be observed by all other devices running typed code. Untyped, compromised devices can still interact with typed devices without being able to violate the policies, except in the case when a policy directly places trust in untyped locations. Importantly, our type system does not need a middleware layer or all users to register with a preexisting PKI, and it allows for devices to dynamically create new identities. The confidentiality property guaranteed by the language is defined for any 
kind of intruder: we consider labeled bisimilarity i.e. an attacker cannot distinguish two scenarios that differ by the change of a protected value. This shows our main result that, for a device that runs well typed code and only places trust in other well typed devices, programming errors cannot cause a data leakage.

\end{abstract}


\section{Introduction}

Organisations commonly trust their cloud providers not to be actively malicious 
but may still need to verify that the cloud service does not make mistakes and does store their data at particular locations.
For example, data protection laws prevent
the storage of certain forms of data outside the European Union.
In order to ensure compliance, a company policy may require that data from a user's device not be
synchronized with a cloud storage provider, unless that provider can certify that the data will not be stored in data centers outside the EU. 
Such checks against inappropriate storage of data can be costly and time consuming, 
sometimes leading to organisations not allowing their employees to use cloud services, 
even though a particular cloud service may be known to be compliant with data handling policies. 


This paper presents a language-based approach to dealing with these kinds of scenarios, of ensuring that data can be shared with cloud services while ensuring compliance with policies on data storage.
The approach is based on a type system that explicitly models trust between cloud services and mobile devices, based on a notion of principals represented at runtime by cryptographic keys.
%
In this language, principals are dynamically associated to (non-disjoint) sets of devices,
and the rights of devices to access data is based on sets of principals (delineating the principals that are allowed to access protected data). 
For instance, if two devices $A$ and $B$ can (individually) act for a principal $P$, while devices $B$ and $C$ can act for principal $Q$,
then the right $\{P,Q\}$ implicitly allows the devices $A$, $B$ and $C$ to access data guarded by this right.
We argue that this two-layer representation of access rights (data guarded by principal rights, and devices acting for principals) is convenient to 
allow principals to share a device, and to allow one principal to use several devices (laptop, mobile phone etc.)

Based on this type system, we present a language that includes most of the primitives necessary for secure imperative programming 
(multi-threading, references, secure channel establishment, and cryptographic ciphers).
A key feature of this language is the ability for new principals to dynamically join the network (in the sense of making network connections to cloud services and other devices) without having to register to any public key infrastructure (PKI) or to use a particular middleware layer\footnote{
  Although middleware for e-commerce (such as CORBA) was popular in the 1990s, that approach was discredited by experience, while SOAP-based approaches have been largely superseded by REST-based Web services, that deliberately eschew the notion of a middleware at least for Internet communication..
}. 

Our threat model assumes some ``honest''
collection of principals (e.g.~the employees of an enterprise) and some collection of devices acting for those principals (e.g.~devices provided to those employees). A device may act for several principals, in the sense that it may issue access requests on behalf of any of the principals that it acts for), while a principal may be associated with several devices.  We describe the devices acting for these principals as \emph{honest devices}, in the sense that they are certified according to the type system presented in this paper to be in conformance with data sharing policies.  
We refer to the corresponding principals for these devices as ``honest'' rather than ``trusted'' because trust management is an orthogonal issue for the scenarios that we consider.  
Our threat model also allows for ``dishonest'' or untyped devices, acting for outside principals, who should not be able to access the data.  These are the attacker devices.
Our security guarantee is that, as long as no honest device provides access to a dishonest principal, the dishonest devices will not be able to obtain any information from any honest devices, unless an honest device has explicitly given the attacker access.

This capability is critical for cloud services.
While it is reasonable to assume that there exists a PKI to certify the identities of cloud providers,
and that cloud providers are trusted by their users,
client devices and their corresponding principals
are unlikely to have certificates.  While an infrastructure for mutual authentication, based on client and server X509 certificates,
 might be provided as part of an enterprise data sharing network, there are difficult issues with extending this trust model to third party cloud service providers.  Furthermore, such a PKI does not provide the level of confidence in conformance with data-sharing protocols, that our approach provides for honest (well-typed) devices.
%
Instead of requiring such a global PKI, the approach described in this article represents principals by cryptographic public keys,
and these are stored on devices like any other values.  Our type system therefore uses a nominal form of dependent types, to reflect these runtime representatives of principal identity to the type system, where access rights on data are tracked through security labels attached to types.

Another contribution of the paper is a security guarantee suitable for our threat model.
%
As stated above, in an open network, we must allow for ``dishonest'' devices that are unchecked (untyped), and potentially malicious.
These devices are able to use any third-party cloud services, including services used by honest devices.
A secure data-sharing system should remain robust to such an intruder.  Our security guarantee is that
if data is protected with rights that only includes honest principals (i.e., that do not include any principal which is associated to an attacker device),
then an attacker cannot learn any information about that data.  In this work, we are focused on confidentiality of the data, and do not consider integrity (that data has not been tampered with by attackers).  There is a notion of integrity in the sense of trust management underlying our approach: Honest principals identify themselves by their public keys, and these keys are to state access restrictions on data, to specify when a device is allowed to ``act for'' a principal, and to validate that communication channels are with honest parties trusted to be type-checked and therefore conformant with data-sharing policies.  This is reflected in several aspects of the communication API, including the establishment of communication channels, generating new principals and transmitting those principals between devices.

A practical difficulty with expressing security guarantees in this setting is that, as principals can be created dynamically, it is not possible to statically check if a variable has a higher or lower security level than some other variable.  Consider the following example:
\begin{example}
 Assume three devices: an honest cloud service with a certified principal $p$, and two (mobile) devices, an honest device $A$ and a malicious device $B$.
 The server code consists of receiving two principal values $p_1$ and $p_2$ from the network before it creates
 two memory locations $x_1$ and $x_2$ where $x_1$ has rights $\{ p , p_1 \}$ while $x_2$ has right $\{ p, p_2\}$.
 A secure location is defined one which cannot be accessed by the attacker.
 If we assume that devices $A$ and $B$ send their principal values (public keys) to the server,
 then depending on which key is received first, either $x_1$ or $x_2$ will be considered secure (only accessible by the honest principal and the cloud provider), while the other
 is explicitly accessible by the attacker.  As usual with information flow control type systems, we propagate these access restrictions through the handling of data by the cloud provider and the devices, ensuring that data protected by the right $\{p,p_i\}$, where $p_i$ is the representative for the honest principal.
\end{example}
Therefore, our security property is a posteriori: once the system has created a memory cell corresponding to some variable,
if the rights associated to this variable at creation time did not include a principal controlled by the attacker, then 
there will never be any leak about the contents of this memory cell to the attacker.
We argue that such a property is suitable for cloud services to increase users' confidence that their data will not be leaked by the service due to a programming error.
Indeed, our system allows us to certify that once some principal creates data, only explicitly authorised principals can obtain information about that data,
by statically checking the code that processes that data.
Verifying the identity of the principals allowed to access the data, and deciding where to place trust among the principals in a distributed system, is an important consideration.  However it properly remains the responsibility of the application written in our language, and a concern which is independent of this type system.  Our approach serves to guarantee proper handling of data among honest principals, once an appropriate trust management system has established who is honest.

Our typed language is intended to be a low level language, without high-level notions such as objects and closures. 
It includes consider references, multi-threading and a realistic application programming interface (API) for distributed communication.
Communications can either be through a secure channel mechanism, implementable using a secure transport layer such as TLS for example,
or through public connections, in which case any device can connect.
The language includes primitives for asymmetric key encryption, since we represent by public keys.
``Possessing'' a secret key, in the sense that the private key of a public-private key is stored in its memory, allows a device to ``act for'' the principal that key represents.  Our approach is similar in philosophy to the Simple Public Key Infrastructure (SPKI), where principals and public keys are considered as synonymous, rather than linking principals to a separate notion of public key representatives.  However, we do not include notions such as delegation that are the central consideration of SPKI, since we explicitly avoid the consideration of trust management, leaving that to applications written using the API that we provide.  This also differentiates our approach from frameworks such as JIF and Fabric, that include delegation of authority to principals based on an assumed trust management infrastructure.  

Nevertheless, there is a notion at least tangentially relegated to delegation of trust in our framework: In order to allow a device to act for more than one principal, our semantics allows a principal to be created on one device and communicated to another device, where it is registered on the receiver device as one of the principals upon whose behalf that device can access data.  For example, a client of a cloud service provider may generate a proxy principal representing that client on the cloud service, and then upload that principal to the cloud service in order to access data that the client is storing on the cloud service.  This ability to share principals across devices is controlled by restrictions established when proxy principals are generated: Such a proxy (client) principal can only be registered on a device that acts for (cloud service) principals that are identified at the point of generation of the proxy.

The security analysis of the type system uses standard techniques from the applied-pi calculus
\cite{AbadiFournet2001}. 
This allows us to prove our correctness property as a non-interference property based on process equivalence, i.e., 
 two systems
differing by one value are indistinguishable by any party that is not allowed to access this value.
The standard pi-calculus includes message-passing with structured values, but does not include an explicit notion of memory (although it can obviously be modeled using processes as references).  Since our language combines message-passing communication and localized stateful memory, we use
the stateful applied pi-calculus \cite{Arapinis2014} as the starting point for our security analysis.
This calculus does not explicitly model location (i.e., the distinction between two processes on the same device and two processes on two distinct devices).
Since this distinction is critical for our security analysis, we add this notion in our calculus.  Nevertheless the proof techniques that we employ are heavily based on those developed for the stateful applied pi-calculus.

The security analysis that we perform expresses that data are secure if 
keys received from other devices are not associated to an attacker.
To formalise this conditional statement, we need more techniques than 
in a standard protocol where data are either secret or public, 
but their status does not depend on the execution.
In our verification in \autoref{sec:annotations}, we introduce an extended syntax that marks
which keys, variables and channels are secure in the current trace.
We then prove that when a new memory location is created with a secure key according to this marking, 
then the attacker cannot distinguish between two scenarios: 
one where the system reduces normally, and another one where the memory location is sometimes altered to another value.   This is the basis for our noninterference property for the security guarantee provided by this approach.  
\iftoggle{techreport}{
}{
The full details of the proofs of correctness for information flow control are provided in the complete version of the paper \cite{tech-report}.
}

In the next section we discuss related work. In \autoref{sec:lang} we present our language, type system and semantics. 
In \autoref{sec:example} we present an extended example 
and in \autoref{sec:result} present our result and outline the proof, 
then we conclude in \autoref{sec:concl}. 
\iftoggle{techreport}{
\autoref{app:other-rules} provides addtional type rules for the language considered in this paper.  \autoref{app:proofs-correctness} considers proofs of correctness omitted from \autoref{sec:result}.  \autoref{sec:ext-rules} considers addtional operational semantics rules for the ``open'' semantics that enables us to reason about interactions with untyped attackers.
}{
}

\section{Related Work}

\paragraph*{Implicit flow}
Implicit information flow properties involve the ability for an attacker to distinguish between two executions.
Previous work that has provided type systems to control implicit information flow \cite{CDLM,Masked08} considered high and low data, 
and this could be extended to a bigger lattice but not to the creation of new principals, 
as the security of a variable is defined statically. 
Zheng and Myers presented an information flow type system that includes dynamic checks on rights \cite{Zheng2005} which can be used, 
for instance, when opening a file.
The Jif Project \cite{JIF} adds security types to a subset of Java, leading to a powerful and expressive language. 
Unlike our work, this other work does not address how to enforce
principal identities and type information 
to be correctly communicated to other locations.

\paragraph*{Security properties on distributed system}
%
%
Work on type security for distributed systems can be distinguished according to the kind of security they aim to provide. 
Muller and Chong present a type system that includes a concept of place \cite{MullerC12} 
and their type system ensures that covert channels between ``places'' cannot leak information. 
Vaughan et al. look at types that can be used to provide evidence based audit  \cite{AURA,JMZ08}. 
Fournet et al. look at adding annotations with a security logic to enforce policies \cite{Fournet2008}. 
Liu and Myers \cite{Liu2014} look at a type system which ensures referential integrity in a distributed setting. 
This work uses a fix lattice of  policy levels, which does not change at runtime. 
The Fabric language \cite{Fabric} provides decentralised, type-enforced security guarantees using a powerful middleware layer for PKI, 
and Morgenstern et al. \cite{Morgenstern2010} extend Agda with security types. 
In contrast, our work allows programs to generate new principals at run-time 
and provides a security property that tracks implicit information flow, without requiring the support of a purpose built middleware layer or global PKI. 
Due to the fact that the attackers in our model can access services in the same way as honest principals, this security property is an adaptation of the bisimulation property which is a strong property introduced in the (s)pi-calculus by \cite{AbadiFournet2001}. 
Bisimilation can be checked for processes by tools like Proverif \cite{Blanchet2001} 
but these kind of tools do not scale up to large systems.


%


%

\paragraph*{Managements of new principals}

Bengtson et al. \cite{Bengtson2011} present a security type system which allows creation of dynamic principals 
in presence of an untyped attacker.  However, this type system  provides only assertion-based security properties of cryptographic protocols.
These  are weaker than non-interference properties as they are expressed on one process instead of comparing two processes.
 \cite{KDLM} considers a framework in which principals can be created at run time (without a global PKI) they prove type soundness rather than a non-interference result. Finally, the DSTAR program \cite{Zeldovich2008} achieves these two goals but is focused on network information 
and relies on local systems to actually analyse implicit flow, which leads to a more coarse system.

\paragraph*{Safety despite compromised principals}
Past work \cite{CD05} has looked at un-typed attackers in a security type system, however this work  only considers a static number of principals, fixed at run time.
Fournet, Gordan and Maffeis \cite{Fournet2007} develop a security type system for a version of the applied pi-calculus extended with locations and annotation. Their type system can enforce complex policies on these annotations, and they show that these policies hold as long as they do not depend on principals that have been  compromised. Unlike our work they assume that the principals are all known to each other and there is a direct mapping from each location to a single principal that controls it. Our work allows principals to be dynamically created, shared between locations and for locations to control multiple principal identities. We argue that this model is a better fit to cloud systems in which users can dynamically create many identities and use them with many services.

\section{Language: Semantics and Type System}
\label{sec:lang}
\begin{figure}
\[
\begin{array}{rcll}
P, Q & ::= &  \nu \activechans. D & \text{devices sharing channels $\activechans$}\\
D & ::= &   \dev{M}{C}  \mid  D & \text{ $C$ running with $M$}\\
& | & 0 & \text{no device}\\
M & ::= & \{ \}  & \text{a memory which maps \dots}\\
& | & \{ x \mapsto v \} \cup M  &\text{\dots a variable}\\
& | & \{ p \mapsto P \} \cup M &\text{\dots a principal name}\\
& | & \{ k' \mapsto \pubKey\}\cup M &\text{\dots a key name}\\
v  & ::= & i & \text{an integer}\\
    & | & \pubKeyOf{k} &  \text{the value of a public key }\\
    & | & \enc{v}{n}{\LR} & \text{a cipher of $v$ with seed $n$}\\
    & | & \enc{\pvalue}{n}{\LR} & \text{an encapsulated principal}\\
    & | & \{v_1, \dots v_n \} & \text{an array of values}\\
    & | & \NotAValue & \text{a special error value}\\
\pvalue & ::= & \prin{\pubKey}{\secKey}{\LR}\!\!\! & \text{a principal value}\\
  \LR & ::= & \{ \pubKey \} \cup \LR  & \text{a set of public keys}\\
 & | & \{~\}\\
 \activechans &::= & \{\} & \\
 &|& \activechans.c &  c :\text{established channel}
\end{array}
\]
\caption{The syntax of devices, values, principals and rights where $\secof{k} / \pubof{k}$ is a secret/public key pair. }
\label{fig:syntax}
\end{figure} 

\subsection{Syntax}
The syntax of our language is given in Figures \ref{fig:syntax}, \ref{fig:types}, \ref{fig:commands} and \ref{fig:expressions}. 
We let $x,y,z$ range over variable names $\VarNames$, $p,p_1,p_2,\ldots$ range over principal names $\PrinNames$, 
$k_1,k_2,\ldots$ range over public key names $\KeyNames$
and $c,c_1,c_2,\ldots$ range over channel names $\ChanNames$.
A system $\nu\activechans.D_1 \mid \ldots \mid D_n$
is a set of \emph{devices} that run in parallel and that communicates through channels of $\activechans$.

The list $\activechans$ records which channel names correspond to establish channels (globally bound).  
Channel also appear in connect and accept commands in the devices, and these are added to $\activechans$ once the channel is opened.
When there are no established channels, we omit the $\nu \{\}.$ prefix.
Note that to guarantee freshness of keys and nonce used in encryption, 
we might also provide global binders $\bar{k}$ and $\bar{n}$ in addition to $\activechans$.
However this guarantee is straightforward to provide, using
``freshness'' predicates, so for readability reasons we omit explicit
binders for generated keys and nonces. 

A device consists of a \emph{memory} $\mem$ and a \emph{command} $C$.
Memories associate variable names with \emph{values}, key names with keys and principal names with \emph{principals}.

Given a nonce $k$ from some assumed set of key nonces, 
we define $(\pubof{k},\secof{k})$ as the public private key pair generated from $k$, where $\pubof{.}$ and $\secof{.}$ are two constructors. 
A principal $\pvalue$ is a tuple $\prin{\secof{k}}{\pubof{k}}{\LR}$ which contains a key pair $(\secof{k},\pubof{k})$, 
together with a (possibly empty) set of public key values $\LR$. 
When a device has $\pvalue$ in its memory, it is allowed to act for $\pvalue$.
Devices that can act as a principal $\pvalue_0$, whose public key $\pubKeyOf{k_0}$ is one of the public keys in $\LR$, 
are allowed to add $\pvalue$ to their memory.

Each variable $x$ in $\VarNames$
represents a reference to a value
i.e. variables are mutable.
At declaration time, a reference is associated with some rights $R$, which cannot be revoked or changed. 
Making all variables mutable is convenient for our security analysis:
it allows us to define non interference as the property that a parallel process 
that alters the value of a high-level variable cannot be detected by an attacker. 
If variables were not mutable we would have to consider a much more complex security property and proof. 
In addition, to avoid to consider scope of variables, we assume that a command never declares twice the same name
for variables, channels, keys and principals.

Types (\autoref{fig:types}) for these variables consists of a pure-type $S$ 
which indicates the base type for the value of the variable (e.g. integer, public key, cipher, etc.)
and a \emph{label} (or \emph{right}) $R$ 
which indicates the principals who are allowed to access to the variable.
A label can be either $\bot$ i.e. the variable is public 
or a set which contains public key names: $k \in \KeyNames$ and $\publicKey{p}$ 
where $p \in \PrinNames$.

Key names are declared by a command $\letk{k}{x} C$.
This command copies the value of the reference $x$ into $k$ which represents
a public key (not a reference to a public key).
The type system \eqref{type:let} ensures that $x$  in this command is  an unrestricted public key: $x : \pubKeyType_\bot$.


Channel types are declared when they are established, we have two kinds of channel: public and secure channels.
Their types have syntax $\Chan(S_{R_1})_{R_2}$ where $S_{R_1}$ is the type of values that are past over the channel 
and $R_2$ expresses 
which principals are allowed to know the existence of $c$ and when communication on this channel takes place.

\begin{figure}
\[
\begin{array}{rcll}
S & ::= & \pubKeyType & \text{a public key} \\
& | & \privKeyType & \text{an encapsulated private key}\\
& | & Int & \text{integers} \\
& | & \cipherType{S} & \text{encrypted data of type }S\\
& | & \arrayType{S} & \text{an array of type }S \\
R & ::= & RS  & \text{a set of rights}\\
  & | & \bot & \text{no restriction}\\
RS & ::= & \{~ \} & \text{the empty set }\\
 & | & \{ K \} \cup RS & \text{$K$ added to a set of rights} \\
K & ::= & k_1,k_2,\ldots,  & \text{ public key names }\\
  & | & \{ \publicKey{p} \} & \text{the public key of a principal}\\

\end{array}
\]
\caption{The types syntax}
 \label{fig:types}
\end{figure}




\begin{figure}

\[
\begin{array}{rcl}
C & ::= & \cond{e_1}{e_2}{C_1}{C_2}
\\& | & \new{x}{S_R}{e};~C  
\\& | & \assign{x}{e};~C 
\\& | & \assign{x[e_1]}{e_2};~C 
\\& | & \letk{k}{x}~C 
\\& | & \paral{C_1}{C_2} 
\\& | &\nop 
\\& | &\bang{C} 
\\& | & \connectPub{c}{\chanType{S_{\bot}}{\bot}};~C 
\\& | & \acceptPub{c}{\chanType{S_{\bot}}{\bot}};~C
\\& | & \connectCCert{c}{\chanType{S_R}{R'}}{k}{p};~C
\\& | & \acceptCCert{c}{\chanType{S_R}{R'}}{k}{p};~C
\\& | & \outputChan{e}{c};~C
\\& | & \inputChanII{c}{x};~C
\\& | &\sync\{G\};C_2 
\\& | & \newPrin{p}{RS}; C 
\\& | & \decrypt{p}{e}{x}{S_{RS}}{C_1}{C_2} 
\\& | & \register{p_1}{e}{p_2}{C_1}{C_2} 
\end{array}
\]
\caption{The syntax of the commands.}
\label{fig:commands}
\end{figure}

\subsection{Semantics}
The semantics of the system is defined as a small-step semantics for commands
and as a big-step semantics for expressions.
Devices run concurrently with synchronized communication between them.
Inside each device, all parallel threads run concurrently and communicate 
through the shared memory of the device (since memory is mutable).
The main reduction rules are presented in \autoref{fig:commands-semantics}
for commands and in \autoref{fig:semantics-rules} for main expressions.
When a command that declares a new variable is reduced, the name is 
replaced by a fresh name that represents a pointer to the location in the memory where the value has been stored.
The evaluation of expressions has the form $\evaluates{\mem}{e}{v}$: 
the evaluation of $e$ with memory $\mem$ returns the value $v$.
We note that this is different from some other calculi, in which variables are not references, and are replaced with a value when declared. Our correctness statement below depends on the use of references and, since we have a memory mechanism, we prefer to store the key names and principal names in memory and instead of applying a substitution,
the names are evaluated when a command reduces 
(cf the \eqref{red:rights} rules). 


Principals are generated using the $\newPrin{p}{RS};C$ command, 
where $RS$ are the keys to use to protect the principal (and therefore cannot $\bot$).
%
%
The rule for this command \eqref{sem:newPrin} generates a fresh key pair $(\pubof{k},\secof{k})$
and stores the principal $\prinv{k}{\mathit{LR}}$ 
at a new location $p'$ in the memory,
where $\evaluates{\mem}{RS}{\mathit{LR}}$.
To bootstrap the creation of these principals, they can be declared with $RS=\{\}$; 
such principal identities can only be used on a single device, they cannot be sent over channels. 
Additionally some devices may start off with the public keys of some trusted parties, 
i.e., the same assumption as TLS. 
This too lends itself well to cloud systems, in which web browsers come with a number of trusted certificates.

Communication between devices uses Java like channels: a channel is first established
then it can be used to input and output values.
The channel establishment is done by substituting the channel names in both devices by a unique fresh channel name  
added to $\activechans$ (we assume that initial channel names and active channel names 
range over distinct sub-domains of $\ChanNames$ to avoid collision).
Note that channels do not name the sending and receiving device as these may be spoofed by an attacker,
however, to get a more tuneable system, it would be a simple extension to add a port number which would restrict possible connections.
For secure channels \eqref{red:open-for}, in a similar way to TLS with client certificates, 
both devices must provide the public key $\pubof{k}$ of who they want to connect to. They must also provide the principal (which includes a private key) to identify themselves to the other party.
%
To set up a secure channel, the client and the server also have to 
ensure that they are considering the same rights for the channel.
For that they have to exchange 
the value of their channel right $\mem(\Chan(S_{R_1})_{R_2})$ 
and make sure that it corresponds to the distant right value $\mem'(\Chan(S_{R'_1})_{R'_2})$.
Indeed, even if type-checking is static inside a device, type-checking 
has to be dynamic between distinct devices since programs are type-checked on each device and not globally.

\begin{figure}
\[
\begin{array}{rcll}
e & ::= & x,y,\dots & \text{variable names}\\
    & | &\publicKey{p} & \text{the public key of $p$}\\ 
    & | & \release{p} & \text{pack a principal}\\
    & | & \encE{e}{RS} & \text{encrypt some data}\\
    & | & e_1 \oplus e_2  & \text{where $\oplus$ is $+,-,\times,\dots$}\\
    & | & \{e_1, \dots e_n \} & \text{an array of expressions}\\
    & | & x[e] & \text{an element of an array} \\
    & | & i & \text{an integer } i \in \mathbb{Z}\\
\end{array}
\]
\caption{The syntax of the expressions }
\label{fig:expressions}
\end{figure}
\newcommand{\front}{\nu \activechans. D ~|~ }
\newcommand{\frontc}{\nu \activechans.c. D ~|~ }
\newcommand{\athread}[1]{\paral{C'}{#1}}
\newcommand{\athreadbis}[1]{\paral{C''}{#1}}

\begin{figure}

\begin{gather}
\inferrule{\text{fresh}(\pubKeyOf{k},\secKeyOf{k}) \quad \text{fresh}(p') \\ M(RS)=\LR}{
\front \dev{M}{\athread{\newPrin{p}{RS};C}}
\\ \rightarrow \front \dev{M \cup \{p' \mapsto \prinv{k}{\LR}\} }{\athread{C\sub{p'}{p}}} }\tag{newPrin\_S}\label{sem:newPrin}
\\
 \inferrule{
 M(e) = v \\
 \fresh(y)
 }{
 \front \dev{M}{\athread{\new{x}{S_R}{e}; ~C}}
\\ \rightarrow \front \dev{M \cup \{y \mapsto v\}}{\athread{C\{y / x \}}}
 }\tag{new\_S} \label{red:new_var}
\\ 
 \inferrule{
M(e) = v_2
 }{
\front \dev{M\cup\{x \mapsto v_1\}}{\athread{\assign{x}{e}; ~C}}
\\ \rightarrow  \front \dev{M \cup \{x \mapsto v_2\}}{\athread{C}}
 }\tag{assign\_S} \label{red:assign}
\\
\inferrule{
\fresh(k') \\ M(e)=\pubKey
}{ \front \dev{M}{\athread{\letk{k}{e} ~C}}
\\ \rightarrow \front \dev{M \cup \{k' \mapsto \pubKey\}}{\athread{C\{k' / k \}}}
}\tag{deref\_S}\label{red:let}
\\
 \inferrule{
M(p) = \prin{k^+}{k^-}{\LR_2} \\
M(e) = \enc{v}{n}{\LR} \\
k^+ \in \LR\\
M(RS)\subseteq \LR \\
\text{fresh}(y)
 }{ 
 \front\! \dev{M\!}{\!\athread{\!\decrypt{p}{\!e}{x}{S_{RS}}{C_1}{C_2}} }
\\  \rightarrow 
 \front \dev{M \cup \{y \mapsto v\}}{\athread{C_1\sub{y}{x}}}
} \tag{dec\_true\_S}\label{red:decT}
\\
\inferrule{  \fresh(c)
}{
 \front \dev{M_1}{\athread{\connectPub{c_1}{\chanType{S_\bot}{\bot}}; C_1}}
\\ ~|~ \dev{M_2}{\athreadbis{\acceptPub{c_2}{\chanType{S_\bot}{\bot}}; C_2}}
\\ \rightarrow \frontc\!
    \dev{M_1}{\athread{C_1 \{c/c_1\}}}
 ~|~ \dev{M_2}{\athreadbis{C_2 \{c/c_2\}}}
 } \tag{open\_public\_S}\label{red:open-pub} 
\\
\inferrule{ \fresh(c) \\
M_1(p_s) = \prin{\secof{k_1}}{\pubof{k_1}}{\LR_s} \\
M_2(p_c) = \prin{\secof{k_2}}{\pubof{k_2}}{\LR_c} \\
M_1(R_1) =  M_2(R_2) \\
M_1(R'_1) =  M_2(R'_2) \\
M_1(k_c) = \pubof{k_2} \\
M_2(k_s) = \pubof{k_1} \\
 }{ 
   \front \\ \dev{M_1}{\athread{\acceptCCert{c_1}{\chanType{S_{R_1}}{R_2}}{k_c}{p_s}; C_1}}
\\   | \dev{M_2}{\athreadbis{\connectCCert{c_2}{\chanType{S_{R'_1}}{R'_2}}{k_s}{p_c}; C_2}}
\\ \rightarrow  
   \frontc \!\dev{M_1}{\athread{C_1 \{c/c_1\}}}
 | \dev{M_2}{\athreadbis{C_2 \{c/c_2\}}}
}
 \tag{open\_priv\_S}  \label{red:open-for}
\\
\inferrule{
 M_1(e) = v  \\
 \text{fresh}(y)
 }{ 
 \front \dev{M_1}{\athread{\outputChan{e}{c};C_1}}
\\ \mid \dev{M_2}{\athreadbis{\inputChanII{c}{x};C_2}}
\\ \rightarrow
 \front \dev{M_1}{\athread{C_1}}
\\ \mid  \dev{M_2 \cup \{ y \mapsto v\}}{\athreadbis{C_2 \{y / x\}}}
 }\tag{i/o\_S} \label{red:i/o}
\\
\inferrule{
M(e) = \enc{\prin{\pubKey_1}{\secKey_1}{\LR_1}}{n}{\LR_1} \\
M(p_2) = \prin{\pubKey_2}{\secKey_2}{\LR_2} \\
     k^+_2 \in \LR_1 \\
     \text{fresh}(p_3)
 }{
 \front \dev{M }{\athread{\register{p_2}{e}{p_1}{C_1}{C_2}}}
 \\ \rightarrow \front \\ \dev{M\cup \{ p_3 \mapsto \prin{\pubKey_1}{\secKey_1}{\LR_1} \}}{\athread{C_1\sub{p_3}{p_1}}}
}\tag{register\_true\_S}\label{red:registerT}
\\
\inferrule{
M(e) = \enc{\prin{\pubKey_1}{\secKey_1}{\LR_1}}{n}{\LR_1} \\
M(p_2) = \prin{\pubKey_2}{\secKey_2}{\LR_2} \\
     k^+_2 \not\in \LR_1
 }{
 \front \dev{M }{\athread{\register{p_2}{e}{p_1}{C_1}{C_2}}}
 \\ \rightarrow \front \dev{M}{\athread{C_2}}
}\tag{register\_false\_S}\label{red:registerF}
\end{gather}
\caption{Main command rules}\label{fig:commands-semantics}
\end{figure}
\begin{figure}
 \begin{gather}
 \inferrule{
 M(p) =\prin{\pubKey}{\secKey}{\LR} \\
 \LR \neq \{\} \\
 \text{fresh}(n)
 }{ 
 M(\release{p}) = \enc{\prin{\pubKey}{\secKey}{\LR}}{n}{\LR}
} \tag{release\_E}\label{rede:release}
\\
 \inferrule{
 \text{fresh}(n) \\
 M(e) = v \\
 M(RS) = \LR
 }{ 
 M(\encE{e}{RS}) = \enc{v}{n}{\LR}
} \tag{enc\_E}\label{red:enc}
\\
  \inferrule{
 M(e_1) = v_1 \\
 M(e_2) = v_2 \\
 v_1 \in \mathbb{Z} \\
 v_2 \in \mathbb{Z} 
 }{ 
 M(e_1+e_2) = v_1 + v_2
} \tag{+\_E}\label{red:sum}
\\
  \inferrule{
 M(e_1) = v_1 \\
 M(e_2) = v_2 \\
 v_1 \notin \mathbb{Z} \vee
 v_2 \notin \mathbb{Z} 
 }{ 
 M(e_1 + e_2) = \NotAValue
} \tag{error\_+\_E}\label{red:sum_err}
 \\
\inferrule{
\forall i, 1 \leq i \leq n, \quad \loc{p_i}{\prinv{k_i}{\LR_i}} \in M
\\ \forall i, n+1 \leq i \leq n+m, \quad \loc{k'_i}{\pubKeyOf{k_i}} \in M
}{ M(\{ \publicKey{p_1},\ldots,\publicKey{p_n}, k'_{n+1},\ldots,k'_{n+m} \})
\\ = \{\pubKeyOf{k_1},\ldots,\pubKeyOf{k_{n+m}}\}
}
\tag{rights\_RS}\label{red:rights}
\end{gather}

\caption{Main semantic rules for expressions and rights}
\label{fig:semantics-rules}
\end{figure}

\begin{example}[Principal set up]
Assume that a cloud service $C$ has a public key $\pubof{c}$, which is known to Alice and Bob,
and that Alice wants to share some private data with Bob using this cloud service.
Alice can do this using the code shown in \autoref{fig:Alice}.
Alice starts by generating a key pair $(\secof{a},\pubof{a})$.
As neither Alice nor Bob have certificates, Bob just sends his key publicly to Alice over a public channel.
Alice receives it, and creates a new variable to be shared with Bob and the cloud service.
She then opens a secure channel with the cloud
that is typed to allow data of type $\{\pubof{a},\pubof{b},\pubof{c}\}$, the $\bot$ right on this channel indicates that, while the data on the channel must be kept confidential, the knowledge that some value has been sent is not.
%
%
This fragment of code does not authenticate Bob, this could be done using another protocol, or offline, but we will show that if the device that sent her this key, and the cloud server, both run well typed code, then she is guaranteed that the secret will only be shared by the device sending the key, the cloud server and herself.
%
She knows that no leak can come from, for instance, bad code design on the cloud device.
\begin{figure}
 \[
 \dev{\loc{k_c}{\pubof{c}}}{
  \begin{array}{l}
    \newPrin{A}{\{\}};\\
  \acceptPub{\mathit{c}}{\chanTypePub{\pubKeyType_{\bot}}};\\
 \inputChanII{c}{x_b};\\
 \letk{k_b}{x_b} \\
 \new{secret}{ \Int_{\{k_c,\publicKey{A},k_b\}}}{42};\\
 \connectPub{upload}{\chanTypePub{\Int_{\{k_c,\publicKey{A},k_b\}}}}\\
 ~~~~~\TO k_c \AS \publicKey{A};\\
 \outputChan{secret}{upload};\\
  \end{array}
   }
 \]
\caption{Alice sharing data with Bob using a cloud service. N.B. this code does not authenticate Bob.}\label{fig:Alice}
\end{figure}

\end{example}

The $\encE{e}{RS}$ expression, governed by the \eqref{red:enc} rule, encrypts the evaluation of $e$ for each of the public keys $\pubof{k_i}$ 
named in $RS$, i.e., 
anyone that has a single private key corresponding to any $\pubof{k_i}$ can decrypt it, 
the set of all $\pubof{k_i}$ is also included in the encryption. 
We use randomised encryption to avoid leakage that would occur otherwise when
 the same value is encrypted twice, and we model this by including a fresh nonce in the encryption.
The $\decrypt{p}{e}{x}{S_{RS}}{C_1}{C_2}$ command reduces successfully  \eqref{red:decT}
when $e$ evaluates to a ciphertext $\enc{v}{n}{RS}$ that can be opened by the secret key of $p$
and that the $\LR$, which is packed into the encryption, is a subset of the evaluation of $RS$. 

The $\mathbf{release}(p)$ expression reduces by encrypting the principal $p$ for each of the set of public keys representing the principals that can access it.
It is the only way to produce a value which contains a secret key and therefore to send private keys through a channel.
The $\mathbf{register}$ command behaves as $\mathbf{decrypt}$ except that it deals with encrypted principals instead of encrypted values.

All other semantics rules are standard except that instead of returning run-time error (division by $0$, illegal offset index etc.)
 expression returns a special value $\NotAValue$.
This feature is critical to guarantee the security of our system.
Indeed, we allow a device to evaluate expressions with secure variables and then to do an output on a public channel.
This scenario is safe only if we can ensure that no expression will block any thread.
Note that the attacker can also send values with some type through a channel of another type,
consequently run time type errors can also occur.

Finally, the command $\mathbf{synchronised}\{C\}$ executes a command $C$ with no communication or interleaving of other processes. This is useful to avoid race conditions.

\newcommand{\wfr}[1]{\GammaPrin; \GammaKey \vdash #1}
\begin{figure}
\begin{gather}
\inferrule{
\wfr{RS}\\ \pc = \bot \\  \pc;\GammaPrin \cup \{ p\}; \GammaKey; \GammaChan;  \Gamma  \vdash C }{  
\pc;\GammaPrin ;\GammaKey; \GammaChan;  \Gamma \vdash \newPrin{p}{RS};C }\tag{newPrin\_T}\label{type:newPrin}
\\
\inferrule{
\wfr{R_1}\\
   \typeVar{e : S_{R_2}} \\
  \rightsComp{R_1}{ \pc \cap R_2}\\
  \exists p \in \GammaPrin .\publicKey{p} \in R_1 \\
  \pc;\GammaPrin;\GammaKey; \GammaChan;  \Gamma \cup \{ x : S_{R_1}\} \vdash C 
}{
\pc;\GammaPrin;\GammaKey; \GammaChan;  \Gamma \vdash \new{x}{S_{R_1}}{e}; C } \tag{new\_T}\label{type:new}
\\
\inferrule{
   \typeVar{x : S_{R_1}} \quad 
   \typeVar{e : S_{R_2}} \\
   \rightsComp{ R_1 }{ \pc \cap R_2 } \quad
  \pc;\GammaPrin; \GammaKey; \GammaChan;  \Gamma \vdash C 
}{
\pc;\GammaPrin;\GammaKey ; \Gamma \vdash  \assign{x}{e}; ~C } \tag{assign\_T}\label{type:assign}
\\
\inferrule{
	\pc = \bot \\
   \typeVar{x : \pubKeyType_{\bot}} \\
   \pc;\GammaPrin; \GammaKey \cup \{ k \}; \GammaChan;  \Gamma \vdash C 
}{
\pc;\GammaPrin; \GammaKey; \GammaChan;  \Gamma \vdash \letk{k}{x} C
}\tag{deref\_T}\label{type:let}
\\
\inferrule{ 
   \typeVar{e_1 : S_{R_1}} \\
   \typeVar{e_2 : S_{R_2}} \\
  \pc \cap R_1 \cap R_2; \GammaPrin; \GammaKey; \GammaChan;\Gamma \vdash C_1 \\
  \pc \cap R_1 \cap R_2; \GammaPrin; \GammaKey; \GammaChan;\Gamma \vdash C_2 
}{ \pc; \GammaPrin; \GammaKey; \GammaChan;  \Gamma \vdash \cond{e_1}{e_2}{C_1}{ C_2}} \tag{if\_T}\label{type:if}
\\
\inferrule{
\wfr{RS_1} \\
 \publicKey{p} \in RS_1 \\
 p \in \GammaPrin \\
\typeVar{e : \cipherType{S}_{R_2}} \\
\rightsComp{RS_1}{(R_2\cap \pc)} \\
  \pc \cap R_2;\GammaPrin; \GammaKey; \GammaChan;  \Gamma \cup \{ x : S_{RS_1}\} \vdash C_1  \\ 
 \pc \cap R_2;\GammaPrin; \GammaKey; \GammaChan;  \Gamma \vdash C_2 
}{
\pc;\GammaPrin; \GammaKey; \GammaChan;  \Gamma \vdash \\ \decrypt{p}{e}{x}{S_{RS_1}}{C_1}{C_2} 
} \tag{dec\_T}\label{type:dec}
\\
\inferrule{
 \pc = \bot \\
  \pc;\GammaPrin; \GammaKey; \GammaChan \cup \{ c \!:\!\!\chanType{S_\bot}{\bot}\}; \Gamma \vdash C 
}{
\pc;\GammaPrin; \GammaKey; \GammaChan;  \Gamma \vdash \connectPub{c}{\chanType{S_\bot}{\bot}}; ~C }   
\tag{connect\_1\_T} \label{type:connectPub}
\\ 
\inferrule{ 
\wfr{R_1}\\ \wfr{R_2}\\
p \in \GammaPrin \\
k \in \GammaKey \\
\rightsComp{  \{\publicKey{p},k\} }{ R_1 \subseteq R_2 \subseteq \pc }   \\
 R_2;\GammaPrin; \GammaKey; \GammaChan \cup \{ c : \chanType{S_{R_1}}{R_2}\};  \Gamma \vdash C 
}{\pc;\GammaPrin; \GammaKey; \GammaChan;  \Gamma \vdash \connectCCert{c}{\chanType{S_{R_1}}{R_2}}{k}{p} ;~ C }
\tag{connect\_2\_T}\label{type:open-cert-client}
\\ 
\inferrule{
   \typeVar{e : S_{R_3}}  \\ 
   c : \chanType{S_{R_1}}{R_2} \in \GammaChan \\
   \pc = R_2 \\
   \rightsComp{R_1}{R_3} \\
  \pc;\GammaPrin; \GammaKey; \GammaChan;  \Gamma \vdash C 
}{
\pc;\GammaPrin; \GammaKey; \GammaChan;  \Gamma \vdash \outputChan{e}{c}; ~C } 
\tag{output\_T}\label{type:output}
\\
\inferrule{
c : \chanType{S_{R_1}}{R_2} \in \GammaChan \\
 R_2 = \pc \\
\pc; \GammaPrin; \GammaKey; \GammaChan;  \Gamma \cup \{ x : S_{R_1}\}  \vdash C 
}{
\pc;\GammaPrin; \GammaKey; \GammaChan;  \Gamma \vdash \inputChanII{c}{x}; ~C } 
\tag{input\_T}\label{type:input} 
\\ 
\inferrule{
         p_2 \in \GammaPrin \\
	\typeVar{e \!:\! \privKeyType_{\bot}} \\
	\pc;\GammaPrin \cup \{ p_1 \}; \GammaKey; \GammaChan;  \Gamma \vdash C_1 \\
	\pc;\GammaPrin; \GammaKey; \GammaChan;  \Gamma \vdash C_2 \\
	\pc = \bot
}{
    \pc;\GammaPrin; \GammaKey; \GammaChan;  \Gamma \vdash\register{p_2}{e}{p_1}{C_1}{C_2} 
} 
\tag{register\_T} \label{type:register} 
\end{gather}
\caption{Typing rules for main commands}
\label{fig:types-rules-cmd}
\end{figure}

\begin{figure}
\begin{gather}
\inferrule{
\wfr{RS}\\
\typeVar{e : S_R} \\
 \rightsComp{ RS }{ R }
}{
\typeVar{ \encE{e}{RS}:\cipherType{S}_{\bot}}} 
\tag{enc\_T}\label{type:enc} 
\\
\inferrule{ p \in \GammaPrin
}{
\typeVar{\release{p} : \privKeyType_\bot}}
\tag{release\_T} \label{type:release} 
\\[2mm]
\inferrule{
\forall i, 1 \leq i \leq n, \quad p_i \in \GammaPrin
\\ \forall i, 1 \leq i \leq m, \quad  k_i \in \GammaKey
}{
\GammaPrin; \GammaKey \vdash \{ \publicKey{p_1},\ldots,\publicKey{p_n}, k_{1},\ldots,k_{m} \}
}
\tag{rights\_T}\label{type:rights}
\end{gather}
\caption{Types rules for non standard expressions and rights}
\label{fig:types-rules-expr}
\end{figure}

\subsection{Types}

The type judgment for expressions takes the form $\GammaPrin ;\GammaKey; \Gamma \vdash e : S_R$
and the type judgment for commands takes the form $\pc; \GammaPrin ;\GammaKey; \GammaChan; \Gamma \vdash C$
where $\Gamma$ is a mapping from variable names to types $S_R$,
$\GammaChan$ from channel names to channel type $\Chan(S_{R_1})_{R_2}$, 
$\GammaPrin$ is a set of principal names,
$\GammaKey$ is a set of public key names,
and where $\pc$ is a right of form $R$ called the \emph{program-counter}. The program counter allows to analyse programs for indirect secure information flow~\cite{Denning1977}.
\autoref{fig:types-rules-cmd} defines the main type rules for commands and \autoref{fig:types-rules-expr} defines the rules for expressions and rights.


In many typing judgments, we use a condition $\rightsComp{R_1}{R_2}$ that states that $R_1$ is more confidential than $R_2$.
The predicate $\rightsComp{R_1}{R_2}$ holds either when $R_2 = \bot$ or when $R_1$ is a syntactical subset of $R_2$
(no aliasing). For instance, we have $\rightsComp{\{k_2,\pub(p)\}}{\{k_1,\pub(p),k_2\}}$ 
and ${\{k_1,k_2\}}\varsubsetneq {\{k_2\}}$ even if $k_1$ and $k_2$ map to the same key in memory.  
We also use $R_1 \cap R_2$ to define the syntactic intersection of the sets $R_1$ and $R_2$ 
(which is $R_2$ if $R_1=\bot$). 

\paragraph*{Types rules for new principals and variables} The typing rule for principal declaration \eqref{type:newPrin} only allows the program counter to be bottom.
 This restriction avoids the situation in which a variable with a right including this principal
might be less confidential than the principal itself.

The new variable declaration \eqref{type:new} checks that the rights $R_1$ to access the new variable $x$
are more restrictive than (a subset of) the rights $R_2$ of the expression being assigned
and of the program counter $\pc$. 
We also ensure that one of the principal in $R_1$ belongs to $\GammaPrin$.
The type rule for assignment \eqref{type:assign} ensures that high security values cannot be assigned to lower security variables.

While the semantics for new principals \eqref{sem:newPrin} stores the rights set $\LR$ dynamically
(as $\LR$ is only used when the principal is sent to another device),
the semantics for managing variables \eqref{red:new_var} does not consider them: 
their confidentiality is entirely provided by the type system.

\begin{example} We consider the following piece of code in which two new principals are declared, 
and both Alice and Bob may know the value of $y$, but only Alice may know the value of $x$:
\[
\begin{array}{ll}
\langle \{\} \blacktriangleright & \newPrin{Alice}{\{\}};\\
& \newPrin{Bob}{\{\}};\\
& \new{x}{\mathit{Int}_{\{\publicKey{Alice}\}}}{5};\\
& \new{y}{\mathit{Int}_{\{\publicKey{Alice},\publicKey{Bob}\}}}{7};\\
&\assign{x}{y};\\
&\mathbf{if}(x\boldsymbol{=}1)~\mathbf{then}~\assign{y}{1};  
\\
\rangle
\end{array}
\]

Here the assignment of $y$ to $x$ should be allowed because $x$ is protected by rights more confidential than $y$.
However, in the last line, the value of $y$ (which Bob can read) 
leaks some information about the value of $x$ (which Bob should not be able to read). 
Therefore, this is an unsafe command and it cannot be typed.
\end{example}

\paragraph*{Types rules for encryption} 
The type rule for encrypting values \eqref{type:enc} verifies 
that the encrypted value is less confidential than the set of keys used for encryption, 
it then types the ciphertext as a public value, i.e., 
encryption removes the type restrictions on a value while ensuring that the encryption provides at least as much protection. 
We note that if the encrypting key depends on non-public data, 
then the \progc would not be public, which would ensured that the ciphertext was not stored in a public variable. 
Hence the use of restricted keys will not leak information.

The corresponding decryption rule \eqref{type:dec}
verifies that the principal $p$ used to decrypt the cipher is valid ($p \in \GammaPrin$)
and is consistent with the rights of the decrypted value.
As the knowledge of which keys has been used to encrypt is protected with the rights $R_2$
of the cipher, the success of the decryption also depends on $R_2$.
Therefore, the \progc has to be at least as high as $R_2$ when typing  the continuation.
Finally, as with an assignment, the rule enforces 
that the created variable does not have a type that is more confidential than the \progc.

\paragraph*{Types rules for public channels} Typing rules for public channels ensures that these are only of type public and, when they are used, the program counter is $\bot$.

\begin{example}
The following system illustrates the use of public channels and encryption:
\[\!\!
\begin{array}{ll}
&\langle \{ Alice \mapsto \prin{k^+_a}{k^-_a}{ \{ \}},bobPub \mapsto k^+_b \}\} \blacktriangleright \\
&~~~~\new{x}{\Int_{\{Alice,bobPub\}}}{7};\\
&~~~~ \mathbf{connect}~c:\chanTypePub{\cipherType{\Int}_{\bot}};\\
&~~~~\outputChan{\encE{x}{\{\pub(Alice),bobPub\}}}{c};\\
&~~~~\inputChanII{c}{e};\\
&~~~~ \decryptThen{Alice}{e}{xInc}{\Int_{\{\pub(Alice),bobPub\}}}\\
&~~~~ \assign{x}{xInc};\rangle\\

|& \langle \{Bob \mapsto \prin{k^+_b}{k^-_b}{ \{ \}},alicePub \mapsto k_a^+ \} \blacktriangleright \\
&~~~~!~\mathbf{accept}~c:\chanTypePub{\cipherType{\Int}_{\bot}};\\
&~~~~\inputChanII{c}{z};\\
&~~~~\decryptThen{Bob}{z}{w}{\Int_{\{\pub(Bob),alicePub\}}}\\
&~~~~\outputChan{\encE{w+1}{\{\pub(Bob),alicePub\}}}{c};\rangle\\
\end{array}\!\!
\]

Alice and Bob start off knowing each other's public keys, and Bob provides a service to add one to any number sent by Alice. 
The variable $x$ is restricted so that only Alice and Bob can know the value. 
Encrypting this value removes these type protections, so that it can be sent across the public channel $c$. 
On Bob's device decrypting with Bob's private key replaces these type protections. 
\end{example}

\paragraph*{Types rules for secure channels}
For secure channels (client and server side), the rule \eqref{type:open-cert-client} enforces that 
the principal who is creating the channel and the principal being connected to, 
 both have the right to access the data passed over the channel, 
hence $\{\pub(p),k\} \subseteq R_1$. 
In order to ensure that the possible side effects caused by using the channel are not more restrictive than the data passed over the channel 
we need that $R_1 \subseteq R_2$. 
The $R_2 \subseteq \pc$ condition stops side channel leakage to the device receiving the connection at the time the channel is opened.
Finally, the \progc is set to $R_2$ once connected: 
this ensures both devices have the same \progc. Without this we would have implicit leakage, for instance, one device with a public \progc could wait for a message from a device with a non public \progc, then outputs something on a insecure channel.
As the sending of the message may depend on a value protected by the \progc,
this would result in a leakage.

We make the strong assumption that the existence of a connection attempt can only be detected 
by someone with the correct private key to match the public key used to set up the connection. 
If we assumed a stronger attacker, who could observe all connection attempts, we will need the condition that $\pc = \bot$
at least for the client. 

The output rule \eqref{type:output} has two main restrictions: 
one which verifies that the device still has the \progc agreed with the corresponding device,
and $R_1 \subseteq R_3$ i.e., 
the type on the channel is no less restrictive than the type of data sent.
This is because when this data is received it will be treated as data of type $R_1$.

For channel creation the restriction on the channel must be at least as restrictive as the program counter. 
For input and output we must additionally check that the program counter has not become more restrictive,
hence requiring that the channel restriction and the program counter are equal, i.e., 
testing the value of a high level piece of data and then sending data over a lower level channel is forbidden.


\begin{example}
 The different roles of $R_1$ and $R_2$ is illustrated in these two programs.
 \[\!\!
\begin{array}{l}
 \new{x}{\Int_{\{\publicKey{Alice},bob\}}}{7};\\
 \acceptCCert{c}{\Chan(\Int_{\{\publicKey{Alice},bob\}})_{\{\publicKey{Alice},bob\}}\\~~~~}{bob}{Alice};\\
 \mathbf{if}(x > 10)\{\\
 ~~\assign{x}{x+1};\\
 ~~\outputChan{x}{c};\}
\end{array}
\]
\[
\begin{array}{l}
 \new{x}{\Int_{\{\publicKey{Alice},bob\}}}{7};\\
 \acceptCCert{c}{\Chan(\Int_{\{\publicKey{Alice},bob\}})_{\bot}}{bob}{Alice};\\
 ~~ \{ \mathbf{if}(x > 10)\{\\
 ~~~~\assign{x}{x+1};~\}\}\\
 ~ | ~
 ~~\outputChan{x}{c};
\end{array}\!\!
\]

Both programs aim at sending $x$ to Bob, which is a secret shared by Alice and Bob.
In the first case, the sending of $x$ depends on its value:
therefore the communication should can only be on a channel with rights 
$\Chan(\Int_{\{\publicKey{Alice},bob\}})_{\{\publicKey{Alice},bob\}}$.
In the other example, even if the value of $x$ is updated due to a parallel thread that has a non public \progc,
the sending of $x$ is unconditional.

Note that the language does not have \emph{``\{if condition then $C$ \}; $C''$} structure as this construct would not be safe:
if $C$ waits infinitely for a connection then $C''$ is not executed. 
However, a \emph{delay} command could be added to
help the second program to output $x$ after $x$ has been updated.
\end{example}

 \paragraph*{Type rules for release and register} The release command is similar to the encryption command except that the rights with which the principal is encrypted are provided by the principal value.
 Therefore, there is no static check to perform in \eqref{type:release}.
 The registration rule \eqref{type:register}, for the same reason has less checks than \eqref{type:dec}.
 However, it does enforce that $\pc=\bot$, without which we could get non public rights; revealing a such a none public right would then be an information leak. Removing this restriction, and allowing non public rights, would be possible in a more complex type system but we decide not to do so to keep the type system more understandable.

\section{Example: A Secure Cloud Server}
\label{sec:example}


\begin{figure}
\small
\[
\begin{array}{l}
Server \equiv \\\new{usage}{ \arrayType{\Int}_{\bot}}{\{0,0,0\}};\\
\new{blocked}{ \arrayType{\Int}_{\bot}}{\{0,0,0\}};\\
\new{nextID}{ \Int_{\bot}}{0};\\
RegisterUsers ~|~CheckUsage
\\[4mm]
RegisterUsers  \equiv \\
!~\acceptPub{\mathit{newUsers}}{\chanTypePub{\arrayType{\pubKeyType}_{\bot}}};\\
~~\inputChanII{newUsers}{client1Client2};\\
~~\letk{client1}{client1Client2[0]} \\
~~\letk{client2}{client1Client2[1]} \\
~~\sync~\{ \\
~~~~\new{accountID}{\Int_\bot}{nextID}; \\
~~~~nextID = nextID+1;\}\\
~~\mathbf{if}~(accountID \leq 2)~\mathbf{then} \{\\
~~~~\new{data}{\Int_{\{\publicKey{Server},client1,client2\}}}{0};\\
~~~~ServeClient(client1,client2,client1)
\\~~|~~ServeClient(client1,client2,client2)~\}
\\[4mm]
CheckUsage \equiv  !~\sync \{\\
~~\new{total}{\Int_\bot\!}{\!usage[0]+usage[1]+usage[2]+3};\\
~~\{~\condnb{usage[0]>total/2}{\\
~~~~\assign{blocked[0]}{1}}{\assign{blocked[0]}{0}};\\
~~|~\condnb{usage[1]>total/2}{\\
~~~~\assign{blocked[1]}{1}}{\assign{blocked[1]}{0}};\\
~~|~\condnb{usage[2]>total/2}{\\
~~~~\assign{blocked[2]}{1}}{\assign{blocked[2]}{0}};\}
\\[4mm]
ServeClient(c1,c2,c3)  \equiv \\
!~\acceptCCert{upload}{\chanTypePub{\Int_{\{Server,c1,c2\}}}\\
~~~~}{c3}{Server};\\
~~\mathbf{if}~(blocked[accountID]=0)~\mathbf{then} \{\\
~~~~\inputChanII{upload}{z};\\
~~~~usage[accountID] = usage[accountID] + 1;\\
~~~~data=z;\}\\
|~ !~\acceptCCert{dowload}{\chanTypePub{\Int_{\{Server,c1,c2\}}}\\
~~~~}{c3}{Server};\\
~~\mathbf{if}~(blocked[accountID]=0)~\mathbf{then} \{\\
~~~~usage[accountID] = usage[accountID] + 1;\\
~~~~\outputChan{data}{download};\}\\
\end{array}
\]
\caption{An example server that monitors the clients usage but not their data}
\label{fig:example}
\end{figure}
\begin{figure}
\[\!
\begin{array}{l}
Sender \equiv \\
 ~~~ \acceptPub{\mathit{otherPrin}}{\chanTypePub{\pubKeyType_{\bot}}};\\
 ~~~ \inputChanII{otherPrin}{mobileKey};\\
 ~~~ \letk{mobile}{mobileKey}\\
 ~~~ \newPrin{Alice}{\{mobile\}};\\
 ~~~ \acceptPub{\mathit{releasedPrin}}{\chanTypePub{\privKeyType_{\bot}}};\\
 ~~~ \outputChan{\release{Alice}}{releasedPrin};\\ 
 ~~~ \acceptPub{\mathit{c}}{\chanTypePub{\pubKeyType_{\bot}}};\\
 ~~~ \inputChanII{c}{bobKey};\\
 ~~~ \outputChan{otherPrin}{bobKey};\\
 ~~~ \letk{bob}{bobKey}\\
 ~~~ \outputChan{\publicKey{Alice}}{c};\\
 ~~~ \connectPub{newUsers}{\chanTypePub{\arrayType{\pubKeyType}_{\bot}}};\\
 ~~~ \outputChan{\{\publicKey{Alice},bobKey\}}{newUsers};\\
 ~~~ Send(Alice,bob,42) \\
 \\
Send(p,k,v) \equiv \\
 ~~~ \connectCCert{upload}{\chanTypePub{\Int_{\{srvKey,\publicKey{p},k\}}}\\~~~~~}{srvKey}{p};\\
 ~~~ \new{sharedSecret}{ \Int_{\{srvKey,\publicKey{p},k\}}}{v};\\
 ~~~ \outputChan{sharedSecret}{upload};\\
\\
Mobile \equiv \\
~~~ \newPrin{Mobile}{\{\}};\\
~~~ \connectPub{\mathit{keyChan}}{\chanTypePub{\pubKeyType_{\bot}}};\\
~~~ \outputChan{\publicKey{Mobile}}{keyChan};\\
~~~ \connectPub{\mathit{releaseChan}}{\chanTypePub{\privKeyType_{\bot}}};\\
~~~ \inputChanII{releaseChan}{encaps}\\
~~~ \registerII{Mobile}{encaps}{MyId}\\
~~~ \inputChanII{keyChan}{bobKey};\\
~~~ \letk{bob}{bobKey} Send(MyId,bob,24)\\
\\
Receiver \equiv\\
 ~~~ \newPrin{Bob}{\{\}};\\
  ~~~ \connectPub{fromBob}{\chanTypePub{\pubKeyType_{\bot}}};\\
  ~~~ \outputChan{\publicKey{Bob}}{fromBob};\\
   ~~~ \inputChanII{fromBob}{aliceKey};\\
   ~~~ \letk{alice}{aliceKey} \\
 ~~~ \connectCCert{download}{\chanTypePub{\Int_{\{srvKey,alice,\publicKey{Bob}\}}}\\~~~~~}{srvKey}{Bob};\\
  ~~~ \inputChanII{download}{data};\\
\end{array}\!
\]
\caption{Definitions of Sender, Receiver and Mobile Processes}
\label{fig:example-2-a}
\end{figure}
\begin{figure}[ht]
\[
\begin{array}{lcl}
SD & \equiv & \langle \{ srvKey \mapsto k^+_s\} \blacktriangleright Sender \rangle \\
MD & \equiv &  \langle \{ srvKey \mapsto k^+_s\} \blacktriangleright Mobile \rangle\\ 
RD & \equiv & \langle \{ srvKey \mapsto k^+_s\} \blacktriangleright Receiver \rangle\\ 
Srv & \equiv & \langle \{Server \mapsto \prin{k^+_s}{k^-_s}{ \{ \}}  \} \blacktriangleright Server \rangle  \\
System & \equiv & SD \mid MD \mid RD \mid RD \mid Srv
\end{array}
\]
\caption{The server context with 4 devices: a sender with its mobile device and two concurrent receivers}
\label{fig:example-2-b}
\end{figure}

As an extended example we consider a cloud server that provides a data storage service. 
The motivation of our work is to make it possible to type an open cloud service, 
without the need for a global PKI neither the need to verify that
its users run typed programs,
so ensuring that it provides security guarantees to all of its users. 
The server process in Figure~\ref{fig:example} defines an open service which users can connect to and register to store data. 
This data can be shared with another principal, hence the server takes a pair of public keys, 
representing the principals, when registering. 

To keep the example simple this server accepts up to 3 accounts and denies further registrations.
The data for each accounts is stored in the data variable defined in the $RegisterUsers$ process; 
the restriction set used to type this variable specifies that only the server and the two clients named at registration 
can have knowledge of this data. 
Additionally, the server keeps track of how often each account is used (in the $usage$ array) 
and runs a process to monitor the usage (the $CheckUsage$ process). 
If any account is found to have made more than 50\% of the total number of requests (plus $3$), 
it is temporarily blocked (by setting the corresponding index in the $blocked$ array to 1). 
The usage data and blocked status are public data. 
This is an example of an open cloud service which writes to public variables after processing private data. 
Our type system ensures that there is no leakage between the two.

An example configuration is given in \autoref{fig:example-2-b}, with the
definitions of the processes provided in
\autoref{fig:example-2-a}: this configuration consists of four devices $SD$, $MD$ and two identical $RD$ devices.
We assume that, in the physical world, 
$SD$ and $MD$ are the laptop and respectively the mobile of Alice while the two other devices $RD$ are owned by Bob and Charlie.
In the system definition, Alice's and Bob's devices start off knowing the servers public key, 
but the server has no knowledge of Alice's and Bob's principals. 
The mobile device $MD$ first creates a new principal identity and shares the public key to $SD$.
Note that $RD$ could also send its private key to $SD$ at this point which is not the expected behavior.
To avoid honest users to establish unwanted connection, a port number mechanism should be added to the connections rules.
Once $SD$ receives the principal's public key from $MD$, 
$SD$ creates a new principal identity to use with the cloud service which is known by the mobile's principal identity.
This allows $SD$ to release and to send the new principal $Alice$ to $MD$ which registers it.
Therefore both $SD$ and $MD$ can use the service with the same account.
Finally Bob's device $RD$ and $SD$ exchange their public keys, 
and $SD$ sends to $MD$ the public key received from $RD$
then $SD$ registers for a shared account between $\pub(Alice)$ and $bobKey$ on the server.
Finally, $SD$ or $MD$ can upload a $sharedSecret$ value to the server.
Meanwhile $MD$ is able to recover the last uploaded value ($0$ if it downloads before an upload occurs).
 
The security type on the variable $sharedSecret$ means that its value can only have an effect on other variables 
with the same or a more restrictive type. 
Importantly, our correctness result limits knowledge of these values to just the Alice, Bob and Server devices, 
no matter what well-typed code are run in these devices. 
On the other hand, checking the authenticity of the Bob key (with a mechanism such as PGP, or out-of-band checks) 
is Alice's responsibility.

These are exactly the guarantees that a user, or a organisation, would want before using the cloud service. 
While many people trust their cloud services, and organisations enter into complex legal agreements, 
leaks can still occur due to programming errors. 
Type checking the code, as we outline in this paper can show that it is safe from such programming errors, 
and help provide the users with the guarantees they need to use the system.





\section{Security analysis}
\label{sec:result}

We now prove that the type system preserves confidentiality of data:
when a variable is declared with rights $R$ then
the only devices that can observe anything when the variable's value changes are the devices that are \emph{allowed} to know one of the keys in $R$.

The proof uses techniques from the applied pi-calculus, rephrased for
our formalism.  Our basic result uses a notion of bisimulation formulated
for reasoning about information flow in
nondeterministic programs \cite{SabelfeldSands00}.
Intuitively, two programs are bisimilar
(for ``low'' observers) if
each one can ``imitate'' the low actions of the other, and at each
step the memories are equivalent at low locations.  Note that memory
can change arbitrarily at each step, reflecting the ability of
concurrent attacker threads to modify memory at any time.  

The applied
pi-calculus extends the well-known pi-calculus, a process calculus
where communication channels may be sent between processes,
with binding of variables to non-channel values.  In our approach,
``memory'' is this set of bindings of variables to values in the
applied pi-calculus.  Also, our bisimilarity is a labelled
bisimilarity since we consider communications on channels as
observable events.  Our
correctness result shows that a high (insider) attacker cannot leak
information to low observers by communication on high channels or by
modifying high locations in memory.

We explain our proof over the following five subsections. In the following subsection we annotate devices with an identifier, so that we can keep track of particular devices as a process reduces, and we define when a particular device is entitled to read data protected with a particular set of rights. In Subsection \ref{sec:attacker} we define our untyped attacker and outline an labelled, open semantics which defines how an ``honest'' (typed) process can interact with an untyped attacker process. We also prove that this open semantics is correct with respect  to the semantics presented above. 

To give us the machinery we need to prove our main results, in Subsection \ref{sec:annotations} we annotate our processes with the rights that apply to all variables. We show that a well annotated, well typed process reduces to a well annotated, well typed process, this results shows that a well typed system does not leak data, but it does not account for untyped attackers. To do this we introduce a labelled bisimulation in Subsection \ref{sec:bisim}. This bisimulation relation defines the attackers view of a process, and their ability to distinguish systems. Finally, in Subsection \ref{sec:results}, we prove that, for our open semantics, a well annotated, well typed process is bisimular to another process that is the same as the first, except that the value of a high level variable is changed. This means that no information can leak about the value of that variable for any possible attacker behaviour, so proving our main correctness results.

\subsection{From rights to allowed devices}
\label{sec:rightsdevices}
As a preliminary step, we need to formally define which devices are and are not granted permissions by a particular set of rights. To do this we need a way to refer to particular devices while they make reductions, so as a notational convention, we place \emph{identifiers} on devices, that are preserved through reductions.  
By convention an identifier will be an index on each device, so for example  $D_1 \mid D_2 \ra D'_1 \mid D'_2$
expresses that $D_i$ and $D'_i$ represent the same physical device in different states.

In Definition~\ref{def:initial-process}, 
Definition~\ref{def:back-rel} and Definition~\ref{def:dev-from} below, we formally define an association between the public keys in a rights set and devices, but first we motivate these definitions with an example:

\begin{example}
 Consider the the system  $SD_A \mid SD_B \mid MD_M \mid MD_N \mid RD_X \mid RD_Y \mid Srv_S$ 
 where $SD,MD,RD$ and $Srv$ are defined in \autoref{fig:example-2-a}
 and \autoref{fig:example-2-b} (i.e. there are two clones of each devices of the system 
 from \autoref{sec:example}, except for the server).
 Consider the variable $\mathit{data} : \Int_{\{\publicKey{\mathit{Server}} , \mathit{client1}, \mathit{client2} \}}$ 
 of the $\mathit{RegisterUser}$ command on Device $S$ (the server).
 There are three reasons for a device to be allowed to access shared
 data, depending on which reduction occurs in the system.
  
  First, the devices that created the keys $\mathit{client1}$ and
  $\mathit{client2}$ are allowed access to this data.  This pair of keys is passed to the
  server at the start of its loop.
 Depending on which device ($A$ or $B$)  made the connection to the server channel $\mathit{newUsers}$, 
 $\mathit{client1}$ allows either Device $A$ or Device $B$ (as $\mathit{client1}$).
 Assume that it is Device $A$.
 Similarly $\mathit{client2}$ represents Device $X$ or $Y$ depending on which device connected to channel $c$
 during the $\mathit{Sender}$ command of Device $A$. Assume it is Device $X$.
 
 Next, since the public key $\mathit{client1}$ has been created by the command $\newPrin{\mathit{Alice}}{\{\mathit{mobile}\}}$ in Device $A$,
 the device which corresponds to the public key $\mathit{mobile}$ is also allowed to access $\mathit{data}$. We assume that it is Device $M$.

 Thirdly, the public key $\publicKey{\mathit{Server}}$ has not been
 generated by any device.  However it was in the initial memory of
 Device $S$, therefore
 this device is also allowed to access the shared data by the right
 granted by this key.
 
 Our security property grants that no other device than Device $S, A, X$ and $M$ can get information about the value of $data$.
 On the other hand, if an untyped attacker provides its own key to Device $A$, through channel $c$,
 then no security guarantee can be provided about
 $\mathit{data}$.
 Indeed, such a case means that the rights explicitly allows the attacker's device to access the data
 as any other regular device.
 
\end{example}
 
Before we formalise what are the allowed devices, we make reasonable assumption about the initial process.
For instance, when the process starts, we assume that devices have not already established any channel between them.
We also consider that they have an empty memory except for some public keys and principals 
(and we do not allow duplicate principals). 
Finally, we consider that all devices are well-typed except one (Device $0$) which is the untyped attacker. 

We first define a well-formed and well-typed condition on processes:

\begin{definition}\label{def:initial-process}
A \emph{valid} initial process 
$P= \nu \activechans.~ \dev{\mem_0}{C_0}_0 \mid \dev{\mem_1}{C_1}_1 \mid \ldots \mid \dev{\mem_n}{C_n}_n$
is a process where:
\begin{enumerate}
 \item There is no active channel already established between the
   processes: $\activechans = \{\}$.
 \item The bound values in memory are either principals or public
   keys: For all $0 \leq i \leq n$, $\loc{x}{w} \in \mem_i$ implies there exists $\pubof{k}$, $w = \prinv{k}{\{\}}$ or $w=\pubof{k}$. 
 \item Each principal exists only on one device: For all $0 \leq i , j \leq n$, $i\neq j$, $\prinv{k}{\{\}} \in \mem_i$ implies $\prinv{k}{\{\}} \notin \mem_j$.
 \item The memory of every device is well-typed with contexts corresponding to its memory and $\pc=\bot$: 
 for all $1 \leq i \leq n$, w.l.o.g. assume that
  $\mem_i =
  \{ \loc{p_1}{\prinv{k_1}{\{\}}},\ldots, \loc{p_m}{\prinv{k_m}{\{\}}},\loc{pk_1}\pubof{k'_1},\ldots \loc{pk_p}{\pubof{k'_p}} \}$, 
  we have
 $\bot; \{p_1,\ldots,p_n\}; \{pk_1,\ldots,pk_p\}; \{\}; \{\} \vdash C$
 for some command $C$.
\end{enumerate}
\end{definition}



Before defining the set of allowed devices, we define an auxiliary function
that maintains which devices are allowed to access shared data, and
which public keys need to be associated to devices.  This auxiliary
metafunction maps backward from a set of rights $\LR$ that confers access, to all
possible devices that may have provided the keys that gave them those
rights.  This is the set of allowed devices $\devicesFrom{\LR}{T}$ where $T$ is a process trace,
defined below.  Any device that is not in this set is not
allowed by $\LR$; we will consider such devices as attacker
devices in our threat model.

\begin{definition}\label{def:back-rel}
 Given a reduction $ P  \ra  P' $ where devices identifiers are in $\{0,\ldots,n\}$, 
given a subset of identifiers $I \subseteq \{1,\ldots,n\}$ and
 a set of public keys $\LR$, 
 we define the backward function $\backward{P \ra P' }(I,\LR)$ in the following way.
 \begin{itemize}
  \item If  $P \ra P' $ is the reduction \eqref{sem:newPrin} 
 on a device $D_i$, $i \neq 0$ that creates a new principal $\prinv{k}{\LR'}$ 
 and that $\pubof{k} \in \LR$ 
 then \[\backward{P \ra P'}(I, \LR)=\left(I \cup \{ i \}, \LR' \cup \LR \setminus \{ \pubof{k} \}\right).\]
 \item Otherwise $\backward{P \ra P'}(I, \LR)=\left(I, \LR \right)$.
  \end{itemize}
  \end{definition}
  
 \begin{definition}\label{def:dev-from}
 Let $P_0 = \dev{\mem_0}{C_0}_0 \mid \dev{\mem_1}{C_1}_1 \mid \ldots \mid \dev{\mem_n}{C_n}_n$ a valid initial process. 
 Let a sequence of reductions $T=P_0 \ra P_1 \ldots \ra P_n $ and let $\LR$ a set of public keys,
 we consider
 \[(I_0, \LR_0)= \backward{P_0  \ra P_1 } 
 \circ \ldots \circ \backward{P_{n-1}\ra P_n}(\emptyset, \LR).\]
 Let $I' = \{i \mid \exists \pubof{k} \in \LR_0, i \in \{1,\ldots,n\}, \prinv{k}{\{\}} \in \mem_i \}$.
 We define the set of allowed devices identifiers $\devicesFrom{\LR}{T}$
 as $\devicesFrom{\LR}{T} = I_0 \cup I'$.
 
Consider the set $\{\pubof{k} \mid \pubof{k}\in \LR_0 \wedge \nexists i \in I', \prinv{k}{\{\}} \in \mem_i \}$.
We say that $\devicesFrom{\LR}{T}$ is \emph{safe} if this set is empty.
\end{definition}

In other words, $\devicesFrom{\LR}{T}$ is safe 
when all keys involved by $\LR$ 
have been either created by devices of $\devicesFrom{\LR}{T}$
or owned by them at the beginning.
This implies, since valid initial processes don't have duplicated keys,
that the untyped attacker whose index cannot be in $\devicesFrom{\LR}{T}$
have not generated any of these keys.

\subsection{Definition of the attacker and of the open process semantics}
\label{sec:attacker}
An attacker is a device $A=\dev{\mem}{C}$ where $\mem$ is a standard memory 
and $C$ is a command which is not typed and which contains
additional expressions to do realistic operations that an attacker can perform 
like extracting $\secof{k}$, $\pubof{k}$ and $\LR$ from $\prinv{k}{\LR}$,
decrypting a ciphertext with only a secret key, or releasing a principal with 
arbitrary rights ($\enc{\prinv{k}{\LR'}}{n}{\LR}$ with $\LR \neq \LR'$).
However, the attacker is not able to create principals with a public key that does not correspond to the private key 
because we assume that the validity of any pairs is checked when received by an honest device.
We denote such an extended expression using $\attExpr$.

To reason about any attacker, we introduce open processes in a similar way as in the applied pi-calculus \cite{AbadiFournet2001}.
An open process has the syntax $\attKnowledge \models \nu\activechans.~ D_{1} \mid \ldots  \mid D_{n}$ 
where $D_1,\ldots,D_n$ are well-typed devices 
(the indexes $1,\ldots n$ are the tags of the devices: we do not change them through reductions),
where $\activechans$ are the channels which have been established between devices $D_1,\ldots,D_n$ (not with the attacker)
and where $\attKnowledge$ is a memory representing the values that the
attacker already received.
We refer to $D_1,\ldots,D_n$ as the \emph{honest devices}.
We also refer to $\attKnowledge$ as the \emph{attacker knowledge}.
This plays the same role as frames in the applied pi-calculus, and
also plays a similar role as computer memory in the bisimulation that
we use for reasoning about noninterference for nondeterministic
programs.
We denote $\attKnowledge(\attExpr)$ the evaluation of 
$\attExpr$ with the memory $\attKnowledge$
(to be defined the variables of $\attExpr$ should exists in $\attKnowledge$).

Our type system ensures that an attacker is never able to learn any of
the secret keys belonging to ``honest'' devices, as represented by the
notion of reference rights defined below.  The following predicate
overapproximates what an attacker can learn about a value, based
on the ``knowledge'' represented by its memory $\attKnowledge$
and on the assumption that it knows all keys which are not in 
$\LR$ (which aims at being the reference rights).

 \begin{definition}\label{def:super-attacker}
Given a set of keys $\LR$ and a value $v$, we define the predicate
$\attKnowledge \getPower_{\LR} v$ as
there exists an extended expression $\attExpr$ such that
$\attKnowledge(\attExpr)=v$, where
this extended expression contains all standard functions and attacker functions,
as well as an oracle function which provides the secret key of any public key which is not in $\LR$.
 \end{definition}

Open processes have two forms of reductions: \emph{internal reductions}, which are the same as the reductions for closed processes, and
\emph{labeled reductions} which are reductions involving the attacker. 
Labels on these latter reductions represent
the knowledge that an attacker can use to distinguish between two
traces, effectively the ``low'' information that is leaked from the system.

There are two forms of labelled reductions, both of which take the form $P \lra{l} P'$. In the first form,  
\emph{input reductions} $P\lra{\LabIn{\ca{c},\attExpr}} P'$,  the attacker provides data, and in the second form, \emph{output reductions} $P\lra{\LabOut{\ca{c},x}} P'$, the attacker receives data from an honest device.
There are also two further forms of input reductions: those for establishing
channels and those which send data.
The reduction that establishes a secure channel takes the form:
\[\!\!
\begin{array}{c}
\attKnowledge\! \grants \!\!\front\!\! \deva{i}{\mem\!}{\!\acceptCCert{c}{\Chan(S_{R_1})_{R_2}\!\!}{k}{P}; C} 
\\  \lra{\LabIn{\ca{c},(\chanType{S_{R_1}}{R_2},\attExpr,\attExpr')}}
 \\  \attKnowledge \grants \front \deva{i}{\mem}{C \{\ca{c}/c\}} 
   \end{array}\!\!
\]
where $\ca{c}$ is any attacker channel name, 
$\attKnowledge(\attExpr)$ should be the private key corresponding to $\mem(k)$
and $\attKnowledge(\attExpr')$ should be the public key of $\mem(p)$.
The reduction to establish a public channel is similar but simpler.
There is no need for checks on $\attExpr$ or $\attExpr'$.

Note that, unlike standard connection establishment where a fresh channel name is 
added to $\nu \activechans$, here the name of the established channel is provided by the attacker
and is not added in $\activechans$, which is out of the scope of the attacker.
However the attacker has to provide channel names in a separate subdomain 
which prevents it from using an existing honest channel name.
The names which are the attacker's channel names are written $\ca{c}$. 
To summarize, a channel name of form $\ca{c}$ represents a channel which is established between a device and the attacker,
a channel $c \in \activechans$ is a channel between two devices (not accessible from the attacker)
and a channel name $c \notin \activechans$ and not in the attacker's channel domain 
is just a program variable representing a future channel. 
Finally, we consider an implicit injection $c \mapsto \ca{c}$ from $\activechans$ to attacker's channels.


For input reductions that sends data on an established attacker channel, 
$\attExpr$ is an expression of the extended syntax
admitting lower level operations that are available to attackers but
not to honest devices, as explained above.
There is just one rule for output reduction, saying that an attacker
can learn from a value output on an attacker channel:
\[
 \inferrule{
 \fresh(x) \\
 \evaluates{\mem_1}{e}{v}
  }{
 \attKnowledge \grants \front \deva{i}{\mem_1}{C' ~|~ \outputChan{e}{\ca{c}}; C }
 \\ \lra{\LabOut{\ca{c},x}} 
 \attKnowledgeAdd{(\loc{x}{v})} \grants \front \deva{i}{\mem_1}{C' ~|~ C}
 } 
\]

\begin{example}\label{example:secure1}
 We consider the system consisting of $\mathit{MD} \mid \mathit{SD}$ from \autoref{fig:example-2-b} 
 running in parallel with an untyped attacker $A$.
 The corresponding open process is initially $\{\} \grants  \mathit{SD}_A \mid \mathit{MD}_{M} $.
 Assume that Device $M$ ($\mathit{MD}_M$) reduces with the attacker instead of $\mathit{SD}_A$, we have:
 \[\begin{array}{ll}
  \ra & \{\} \grants \mathit{SD}_{A} \mid \mathit{MD}'_{M}  \\
  \lra{\LabIn{\ca{c}, ct}} & \{\} \grants \mathit{SD}_{A} \mid \mathit{MD}''_{M}  \\
     \lra{\LabOut{\ca{c},x}} & (\loc{x}{k^+}) \grants \mathit{SD}_{A} \mid \mathit{MD}'''_{M} 
   \end{array}
\]
where $ct = \chanTypePub{\pubKeyType_{\bot}}$ and Device $M$ has memory $\loc{\mathit{Mobile}'}{\prinv{k}{\{\}}}$
(a renaming of the local variable $\mathit{Mobile}$ on the device),
after the first internal reduction where the principal is created.

After the first reduction where $\mathit{MD}_M$ creates its principal,
the attacker establish the connection with Device $M$: 
as $M$ is expecting a connection of type $ct$, 
the attacker have to provide $ct$ and one of its channel name $\ca{c}$.
Next, Device $M$ outputs the value of $\pub(Mobile)$ on $\ca{c}$ which is then stored on the attacker's memory. 
\end{example}


The following defines a subprocess of the ``honest'' devices of
a system $P$, where some (other) devices of that system may be attacker devices.
This subprocess of honest devices will be those defined by
$\devicesFrom{\LR}{T}$.  In other words, all devices
which are not allowed by some key in $\LR$ are assumed to be controlled by the attacker.


\begin{definition}
 Given a process $P=\nu \activechans.~ D_{1}~|~ \dots ~|~ D_{n} $ 
 and $\{\ident_1,\dots,\ident_m \} \subseteq \{1,\dots,n\}$, 
 let $\activechans'$ the names of channels between devices of $\{\ident_1,\dots,\ident_m \}$
 we define the subprocess of devices
 $\subp{P}{\{\ident_1,\dots,\ident_m\}} = \nu \activechans'. D'_{\ident_1}~|~ \dots ~|~ D'_{\ident_m}$ 
 where $D'_i$ is $D_i$ where each channel name $c \in \activechans \setminus \activechans'$
 have been replaced by an attacker-channel name $\ca{c}$.
\end{definition}

In the following, we will denote by $P\Lra{l_1,\ldots,l_n} P'$ a
sequence of reductions $P \ra^* P_1 \lra{l_1} P'_1 \ra^* P_2 \ldots P_n \lra{l_n} P'_n \ra^* P'$. 
We also denote by $P\lra{l^?} P'$ a reduction which is either $P \ra P'$ or $P \lra{l'} P'$ for some label $l'$ 
and by $P\lra{l^?}\!^? P'$ either $P\lra{l^?} P'$ or $P=P'$.

The following proposition states that if there is an execution 
of a system that includes communication with attacker devices, where
all communication is local in this closed system, then we can consider
subsystem of this, omitting the attacker devices,
 where communications with attacker devices are
modelled by labelled reductions of the form described above.  So the
attacker devices become part of the context that the devices in the
subsystem interact with through a labelled transition system.  Such a
subsystem may also exclude some of the ``honest'' devices in the
original system, and in that case we treat those excluded devices as
attacker devices in the context (since labels on the reductions only
model communication with attackers).



\begin{restatable}{proposition}{subprocess}\label{prop:identifiers}
Let $P=\nu \activechans. A \mid D_{1} \mid \ldots \mid D_{n}$ where all $D_{i}$ are well-typed and $A$ is an untyped device.
If  $ P \ra\!^* \nu \activechans'. P'$,    
then for all subset $S$ of $\{1, \dots, n\}$, 
there exist $l_1,\ldots,l_m$ and $\attKnowledge'$ (the attacker
knowledge at the end of the execution) such that 
\[\attKnowledge \grants \subp{P}{S}
\Lra{l_1,\ldots,l_m} 
\attKnowledge' \grants \subp{P'}{S}\]
where $\attKnowledge$ is 
a memory which is the union of the memories of $A$ and all $D_i$ with $i \notin S$ 
(variable names are renamed whenever there is a name conflict).
\end{restatable}


For instance, the system $SD_A \mid \mathit{MS}_M \mid \mathit{SD}_0$,
where $\mathit{SD}_0$ is part of the attacker,
can perform three internal reductions between $\mathit{MS}_M$ and $\mathit{SD}_0$
where $\mathit{MS}_M$ sends its public key to $\mathit{SD}_0$.
In \autoref{example:secure1}, we provided the three reductions of the system 
$\subp{(\{\}\grants \mathit{SD}_A \mid \mathit{MS}_M \mid \mathit{SD}_0)}{\{A,M\}}$.

\subsection{Extended syntax with extra annotations}
\label{sec:annotations}
In this section, we add extra annotations to processes to perform a specific analysis about a given right $\LR=\{pk_1,\ldots,pk_n\}$.
Since the keys in $\LR$ do not necessary exist in the initial process,
we first annotate open processes with a \emph{reference right} $\refrights$
with the intention that this right will eventually grow to
the right $\LR$ as keys are generated and added to this right during execution.
Given a reference right $\refrights$, 
we define any rights whose all keys are in $\refrights$ to be \emph{\highr} (by opposition to \emph{\lowr}).
The set $\refrights$ starts with keys that exists in the initial
process.

An \emph{annotated process} has the syntax 
\[\attKnowledge; \refrights \models P\]
for attacker knowledge $\attKnowledge$, reference right $\refrights$
and process $P$.  The exact form of the annotations is
explained below.
The reference right $\refrights$ only changes during a reduction of a command $\newPrin{p}{RS}$.
In this case, there is a choice of whether or not to include the
generated public key in $\refrights$. 

\begin{definition}
A sequence of reductions is called \emph{\standard} if 
each time a \eqref{sem:newPrin} reduction adds a key to $\refrights$:
\begin{multline*} \attKnowledge; \refrights \models \nu \activechans. \dev{\mem}{\newPrin{p}{RS}; C} \mid P \\ \ra 
\attKnowledge; \refrights \cup \{ \pubof{k}\} \models \nu \activechans. 
\dev{\mem \cup{
\loc{p}{\prinv{k}{\LR}}}}{ C} \mid P \!\!\! \end{multline*}
we have $\LR \subseteq \refrights$.
\end{definition}


The next proposition ensures the existence of a \standard annotation such that
the rights we want to consider at the end are \highr and that the attacker does know any key 
in the set $\refrights$ of the initial process.
We will state in  \autoref{lemma:subred} that this implies that the attacker never knows the keys in
$\refrights$ during the whole reduction.

\begin{restatable}{proposition}{safeStandard}\label{prop:safe->standard}
Let $P$ be a valid initial process such that $P \ra{\!^*} P'$ and let $\LR$ a set of keys defined in $P'$.
If $D_H=\devicesFrom{\LR}{P \ra{\!^*} P'}$ is safe
then there exists a \standard annotation for the reduction
provided by \autoref{prop:identifiers}:
$\attKnowledge_0; \refrights_0 \grants \subp{P}{D_H} \Lra{l_1,\ldots,l_n} \attKnowledge; \refrights \grants \subp{P'}{D_H}$
where $\refrights_0$ and $\refrights$ are
such that $\forall \pubof{k} \in \refrights_0,~ \attKnowledge_0 \ngetPower_{\refrights_0} \secof{k}$ 
and $\LR \subseteq \refrights$.
\end{restatable}

When a variable or a channel is created with some rights, the
reduction rules of the operational semantics 
remove all information about those rights.  Therefore we add annotations to devices
to remember if the defined right was \highr or \lowr, according to
$\refrights$.  Recall that a right is \highr if it
contains a key in the reference right $\refrights$.
We use $\ell$ as a metavariable for an annotation $\high$ or $\low$.
\begin{enumerate}
 \item A memory location $\loc{x}{v} \in \mem$ which is created by a command $\new{x}{S_R}{e};C $,
 where $R$ evaluates in $\mem$ to a \highr right according to $\refrights$, is annotated as a \highr location: $\loch{x}{v}$.
 Otherwise $x$ is annotated as a \lowr location: $\locl{x}{v}$. By convention $\attKnowledge$ contains only \lowr locations.
 \item 
A new channel name $c$, which is created by a command that establishes a channel $c : \Chan(S_{R_1})_{R_2}$,
 is annotated as ${c}^{\ell_1}_{\ell_2}$ where $\ell_1$ resp. $\ell_2$ is $\high$ if the evaluation of $R_1$
 (resp. $R_2$) is $\highr$, and is annotated as $\low$ otherwise.
 \item  A value which is the result of an expression (besides an
   encryption expression) where one variable refers to a \highr{} location
 is annotated as a \highv value $\tagged{v}$. Otherwise, it is annotated as a \lowv{} value $\untagged{v}$. 
 When a value $\tagx{v}{\ell}$ is encrypted, it becomes $\untagged{\enc{\tagx{v}{\ell}}{n}{RS}}$ 
 i.e., the initial tag is associated to the subterm.
 \item A device $\dev{\mem}{C_1 ~|~ \dots ~|~ C_n}$ 
 is annotated as $\dev{\mem}{\thr{\ell_1}{C_1} ~|~ \dots ~|~ \thr{\ell_n}{C_n}}_i$.
There is one tag $\ell_i \in \{\low,\high\}$ for each sub-command $C_i$ which can be reduced.
The annotation of each {thread} $\thr{\ell_i}{C_i}$ records whether 
the existence of this thread was due to a \highv value or a \highr right.
When $\ell_i=\high$, we say that the thread is \emph{\hight}, otherwise the thread is \emph{\lowt}. 
For instance, the annotated device $\dev{\loch{a}{\tagged{1}}}{\thr{\low}{\IF\ (a = 0) \THEN C_1 \ELSE C_2}}_1$
reduces to $\dev{\loch{a}{1}}{\thr{\high}{C_2}}_1$ because $a$ is a \highr location.
The other case where a thread can be set to \hight is in the
establishment of a secure channel that is annotated with  $\chanhh{c}$.
\end{enumerate}


These annotations allow us to define technical invariants that are preserved during reduction.
In the following technical definition, we formalize the idea that 
there exist a typing judgment for devices (Case~\eqref{safe:well-type}
below) which is consistent with 
the annotations: 
\begin{enumerate}
\item
Variables (Case~\eqref{safe:variables}) and channels (Case~\eqref{safe:channel}) 
should have a type which corresponds to their annotation
\item The type system tracks a notion of security level for the
  control flow, and this level $\pc$ must be consistent with the
  thread annotation (Case~\eqref{safe:pc}).
\item
In addition, we express that the devices are not in a corrupted
configuration, in the sense that
secure channels are not being used to communicate with the attacker
(Case~\eqref{safe:no-high-attacker-chan}), nor are they used 
in a \lowt thread (Case~\eqref{safe:low-pc-low-chan}).
\item
Finally, ciphers for values (Case~\eqref{safe:cipher-variable}) and principals (Case~\eqref{safe:released-prin})
should not have \highv contents protected by \lowr keys. 
\end{enumerate}

\begin{definition}\label{def:safe}
 A tuple consisting of a reference right $\refrights$, 
 a memory $\mem$ and a thread $\thr{\ell}{C}$ is \emph{\wa} 
 written ``\emph{$\safe{\refrights}{\ell}{\mem}{C}$}'' if
 there exists $\pc$, $\GammaPrin$, $\GammaKey$, $\GammaChan$ and $\ctx$ such that
 \begin{enumerate}
  \item \label{safe:well-type} $\pc;\GammaPrin;\GammaKey;\GammaChan;\ctx \vdash C$
  \item \label{safe:variables} For all locations $x$ in $\mem$, we have 
  \begin{itemize}
   \item either $\locl{x}{\untagged{v}}$ in $\mem$ for some value $v$ and $x \in \GammaPrin \cup \GammaKey$
  \item or $\GammaPrin;\GammaKey;\ctx \vdash x : S_R$ and
  \begin{itemize}
   \item if $\refrights \lc \mem(R)$, $\loch{x}{\tagged{v}}$ in $\mem$, 
   \item if $\refrights \nlc \mem(R)$, $\locl{x}{\untagged{v}}$ in $\mem$ for some value $v$. 
  \end{itemize}
    \end{itemize}
  \item \label{safe:channel} For all channels $c$ in the thread such that $\GammaChan \vdash c : \Chan({S}_{R_1})_{R_2}$, 
  the annotation of $c$ is $c^{\ell_2}_{(\ell_1)}$ 
  where $\ell_2=\high$ iff. $\refrights \lc R_2$, $\ell_1=\high$ iff. $\refrights \lc \mem(R_1)$.
  \item \label{safe:pc} $\refrights \lc \pc$ if and only if $\ell=\high$.
  \item \label{safe:no-high-attacker-chan} For all $\ca{c}$, 
  we have $\GammaChan \vdash \ca{c} : \Chan({S}_{R_1})_{R_2}$ 
  with $\refrights \nlc \mem(R_2)$ and $\refrights \nlc \mem(R_1)$.
 \item \label{safe:low-pc-low-chan}if $\ell$ is $\low$, then there is no $\chanhh{c}$ in $C$.
  \item \label{safe:cipher-variable} For all values $v$ stored in memory, if a sub-term of $v$ matches $\enc{\tagged{t}}{n}{\LR}$
  then $\refrights \lc \LR$.
  \item \label{safe:released-prin} For all values $v$ in memory, if a sub-term of $v$ matches $\enc{\prin{k}{\LR}}{n}{\LR'}$
  then $\LR = \LR'$ or $\refrights \lc \LR$.
 \end{enumerate}
 
 When, for a device $D=\dev{\mem}{\thr{\ell_1}{C_1}~|~ \dots ~|~ \thr{\ell_n}{C_n}}_i$ and a set $\refrights$,
 we have $\safe{\refrights}{\ell_i}{\mem}{C_i}$ for $i\in\{1,\dots,n\}$,
 then we use the notation $\safeDev{\refrights}{D}$.
\end{definition}

Finally, we get the following subject-reduction result:

\begin{restatable}{proposition}{subjectreduction}\label{lemma:subred}
Let $\attKnowledge; \refrights \grants \nu \activechans. D_1 \mid \ldots \mid D_n$ an open process
such that for all $\pubof{k} \in \refrights$ we have $\attKnowledge \ngetPower_\refrights \secof{k}$
and for all $1 \leq i \leq n$ we have $\safeDev{\refrights}{D_i}$.
If 
 $\attKnowledge; \refrights \grants \nu \activechans . D_1 \mid \ldots \mid D_n
 \lra{l^?} 
 \attKnowledge'; {\refrights}' \grants \nu \activechans' . D'_1 \mid \ldots \mid D'_n  $
 then 
 for all $1 \leq i \leq n$ we have $\safeDev{\refrights}{D'_i}$.
 Moreover if the reduction is \standard,
 we have $\attKnowledge' \ngetPower_{\refrights'} \secof{k}$ for all
$\pubof{k} \in \refrights'$.
\end{restatable}

Finally, valid initial processes which only require each device to be well-typed are \wa.

\begin{restatable}{proposition}{initialprocess}\label{prop:init-well-annotated}
 Given a valid initial process $P = \Big(\attKnowledge \grants \deva{1}{\mem_1}{C_1} | \ldots | \deva{n}{\mem_n}{C_n} \Big)$, 
 We have $\safe{\refrights}{\low}{\mem_i}{C_i}$ for $1\leq i\leq n$ for any $\refrights$.
\end{restatable}


\subsection{Labelled bisimilarity}
\label{sec:bisim}

The invariants expressed above are about a single process. 
To ensure that the attacker cannot track implicit flows, 
we need to compare the execution of two processes in parallel.
In this section, we define a relation between processes which implies
an adapted version of the bisimilarity property of the applied pi-calculus~\cite{AbadiFournet2001}.

The two processes that we compare are the actual process and another one 
where the value of one of the \highv variables of the memory of some device has been changed to another one.
So first, we define what is a process where a variable is modified arbitrarily.

\begin{definition}
 Given a device $D=\dev{\mem \addm{\{\loc{x}{v'}\}}}{C}$, and a value $v$,
 we define $D_{x=v}$ to be the device that updates the variable $x$ to
 be $v$ by assignment, $\dev{\mem
   \addm{\{\loc{x}{v'}\}}}{\paral{C}{\assign{x}{v};}}$.
 Extending this from devices to processes, given a process $P$ which contains $D$ with index $i$, 
we define $P_{i:x=v}$ to be the same process
 except that $D$ has been replaced by $D_{x=v}$.
\end{definition}

In contrast to systems where the attacker can only observe public
values in memory after reduction,
here we model that the attacker can observe communications on channels
that have been established with other devices.
In the following examples, we stress how an attacker can distinguish between two processes even without 
knowing actual confidential values in the memory of those devices. 

\begin{example}
We consider the device $D(X,Y)$, where the description is parametrized
by two meta-variables: $X$ is a value stored in memory and $Y$ is a
value that is encrypted and sent on an attacker channel.  The
device has a memory $\mem(X) = \{\loc{x}{X} , \loc{k}{\pubof{k_0}}\}$,
and the full description of the device is
\[
  \devA{\mem(X)}{ \IF \ x = 0 \THEN \outputChan{\encE{Y}{k}}{\ca{c}}}.
\]

The process $D(0,Y)$ can be distinguished from 
the process $D(1,Y)$.
In the first case, we have:
\begin{eqnarray*}
\lefteqn{ \{\} \grants D(0,Y) \lra{\LabOut{\ca{c},m}} } \\
& &
\{\loc{m}{\encE{Y}{k}}\} \grants \devA{\mem(X)}{\SKIP}
\end{eqnarray*}
On the other hand, there is no reduction with only the label
$\LabOut{\ca{c},m}$ and internal reductions
starting from $\{\} \grants D(1,Y)$.
This distinction models the fact that if the attacker receives data on
$\ca{c}$, it learns that $X=0$.

The processes 
$D(0,2)$ and $D(0,3)$ can also be distinguished
even if both 
$\{ \loc{k}{\secof{k_0}} \} \grants D(0,2)$ and $\{  \loc{k}{\secof{k_0}}  \} \grants D(0,3)$
can reduce with a label $\LabOut{\ca{c},m}$.
Indeed, after reduction the attacker's knowledge is $\attKnowledge = \{\loc{k}{\secof{k_0}},\loc{m}{\enc{Y}{n}{\pubof{k_0}}}\}$
where $Y$ is $2$ (resp. $3$): 
By performing the decryption of $m$ with an untyped decryption, the
attacker can compare the result to $2$, and  
$\attKnowledge(\mathrm{decr}_k(m))=2$ is only true in the first case.

Finally, the processes $\{\} \grants D(0,2)$ and $\{\} \grants D(0,3)$ are not distinguishable.
In both cases:
\begin{itemize}
 \item there is no test to distinguish between the two attacker's
   knowledge, and
 \item the labelled reductions are the same i.e $\LabOut{\ca{c},m}$
\end{itemize}
\end{example}

With these examples in mind, 
we introduce static equivalence, an adaptation of the one used in the applied pi-calculus, which
expresses that it is not possible to test some equality which would work with one memory but not with the other.

\begin{definition}
 We say that two memories $\mem_1$ and $\mem_2$ are statically equivalent ($\approx_s$) if 
they have exactly the same variable names and for all 
 extended expressions $\attExpr$ where its variables $x_1,\ldots,x_n \in
 \mem_1$, we have $\mem_1(\attExpr)=\mem_2(\attExpr)$.
\end{definition}

Finally, we define an adaptation of the labelled bisimilarity.
This recursive definition generalizes the conclusion of the example:
the two memories should be statically equivalent,
and a transition in one process can be mimicked in the other process,
where the reduced processes should also be bisimilar.

\begin{definition}
 Labeled bisimilarity ($\approx_l$) is the largest symmetric relation $\mathcal{R}$
on open processes such that $(\attKnowledge^A \grants P^A) \labrel (\attKnowledge^B \grants P^B)$ implies:
\begin{enumerate}
 \item $\attKnowledge^A \approx_s \attKnowledge^B$;
 \item if $(\attKnowledge^A \grants P^A) \ra ({\attKnowledge^A}' \grants {P^A}')$, 
 then $(\attKnowledge^B \grants P^B) \ra^{\!*} ({\attKnowledge^B}' \grants {P^B}')$ 
  and $({\attKnowledge^A}' \grants {P^A}')\labrel({\attKnowledge^B}' \grants {P^B}')$ for some ${\attKnowledge^B}' \grants {P^B}'$;
 \item if $(\attKnowledge^A \grants P^A) \lra{l}  ({\attKnowledge^A}' \grants {P^A}')$, 
 then $(\attKnowledge^B \grants ({\attKnowledge^B} \grants {P^B}) \Lra{l}  
 ({\attKnowledge^B}' \grants {P^B}')$ and $({\attKnowledge^A}' \grants {P^A}')\labrel({\attKnowledge^B}' \grants {P^B}')$ 
 for some ${\attKnowledge^B}' \grants {P^B}'$.
\end{enumerate}
\end{definition}

The labelled bisimilarity is a strong equivalence property, and its
we use it to state the security property that an attacker is unable to distinguish
between two processes in the same class \cite{Arapinis2014}. 
However this definition does not help to actually compute the processes in the same class.
Therefore we define a stronger relation
that can be defined from our annotated semantics.

First, we define an obfuscation function $\erase{\LR}{w}$ 
which takes a set of public key $\LR$ 
and a value or a principal $w$ 
and returns an
obfuscated value (its syntax is like the syntax of value except that there is an additional option $\erasedValue$
and that a principal value is also an option).


 \begin{definition}\label{def:erase} 
 Let $\LR$ a set of public key, $w$ a principal $p$ or a value $v$.
  We define $\erase{\LR}{w}$ depending on the structure of $w$.\\
  {\small
  \[\!\!\begin{array}{lcl}
  \erase{\LR}{\enc{w'}{n}{\LR'}}\!\!&\!\!\!=\!\!\! & \!\!\! \begin{cases}
   \enc{\erase{\LR}{w'}}{n}{\LR'} \text{ if }\LR \nlc \LR' \\
   \enc{\erasedValue}{n}{\LR'} \text{ if } \LR \lc \LR'   
                               \end{cases}\\[4mm]
  \erase{\LR}{\{v_1,\ldots,v_n\}}\!\!&\!\!\!=\! \!\!& \!\!\! \{ \erase{\LR}{v_1},\ldots, \erase{\LR}{v_n} \}\\[1mm]
  \text{otherwise }\erase{\LR}{w}\!\!&\!\!\!=\!\!\! & \!\!\! w.
  \end{array}\!\!\!\!\!\!\!\!\!
  \]}
  where $w'$ is a value or a principal and $n$ a nonce.
 \end{definition}

The set $\LR$ aims at containing a set of keys whose private keys will never be known by the attacker.
In our previous example, we have $\erase{\pubof{k_0}}{\enc{Y}{n}{\pubof{k_0}}}=\enc{\erasedValue}{n}{\pubof{k_0}}$:
if $\secof{k_0}$ is never known by the attacker, the attacker will never be able to know anything about $Y$.
Note that the random seed $n$ used for the cipher is not hidden.
Indeed, the attacker is able to distinguish between a cipher that is sent twice
and two values which are encrypted then sent:
in the first case the two message are strictly identical.

We now define an equivalence on memories: 
given a reference right $\refrights$, a set of safe keys, then we say
that two memories are equivalent 
if they differ only by the \highv values which aims at never been sent to the attacker
and by the obfuscated terms of the \lowv values.

\begin{definition}[equivalent memories]\label{def:eqmem}
\emph{Equivalent memory} is a relation on memories such that $\mem_1
\memrel{\refrights} \mem_2$ 
if 
  for each \lowr location $\locl{x}{v^A}$ in $\mem_1$ (resp. $\mem_2$), 
  there exists $\locl{x}{v^B}$ in $\mem_2$ (resp. $\mem_1$) such that
  $\erase{\refrights}{v^A}=\erase{\refrights}{v^B}$ and each location $\loc{p}{\pvalue}$ in $\mem_1$ (resp. $\mem_2$) 
  exists in $\mem_2$ (resp. $\mem_1$).
\end{definition}

Next, we define an equivalence relation between two \wa processes:
two processes are equivalent if they have the same commands up to some additional \hight threads
and have equivalent memories:

\begin{definition}\label{def:final-relation}
Two annotated processes $P^A$ and $P^B$ are \emph{annotated-equivalent}  ($P^A \labrel P^B$)  
when there exists an alpha-renaming (capture-avoiding renaming of bound variables) of $P^B$ such that
\begin{multline*}
P^A\ \mathrm{is}\  \attKnowledge^A; {\refrights} \grants \nu \activechans. \deva{1}{\mem^A_1}{\paral{C^A_{(1,1)}}{\paral{\dots}{C^A_{(1,m_1)}}}} | 
\\
\dots | \deva{n}{\mem^A_n}{\paral{C^A_{(n,1)}}{\paral{\dots}{C^A_{(n,m_n)}}}},
\end{multline*}
\vspace{-2em}
\begin{multline*}P^B\ \mathrm{is}\ \attKnowledge^B;{\refrights} \grants \nu \activechans'. \deva{1}{\mem^B_1}{\paral{C^B_{(1,1)}}{\paral{\dots}{C^B_{(1,m_1)}}}} | 
\\ 
\dots | \deva{n}{\mem^B_n}{\paral{C^B_{(n,1)}}{\paral{\dots}{C^B_{(n,m_n)}}}},
\end{multline*}
 and we have $\attKnowledge^A \memrel{\refrights} \attKnowledge^B$,
 for all $\pubof{k} \in \refrights$, we have $\attKnowledge^A \ngetPower_\refrights \secof{k}$
 and  furthermore for all $(i,j)$,
 first
 $\mem^A_i \memrel{\refrights} \mem^B_i$,
  next,
  the annotation on each thread $C^X_{(i,j)}$ (where $X$ stands for $A$ or $B$)
  is either $\high$ or $\low$ such that either:
   \begin{itemize}
    \item $\safe{\refrights}{\high}{\mem^X_i}{C^X_{(i,j)}}$
    \item $\safe{\refrights}{\low}{\mem^X_i}{C^X_{(i,j)}}$ and
      $C^A_{(i,j)} = C^B_{(i,j)}$ (the commands are syntactically
      identical up to renaming of bound variables).
   \end{itemize}
\end{definition}

Finally, we prove that this relation actually implies bisimilarity.

\begin{restatable}{proposition}{implylabelled}\label{prop:=>labelled}
Let $P^A$ be $(\attKnowledge^A; \refrights \grants \nu \activechans.  Q^A)$ 
and $P^B$ be $(\attKnowledge^B; \refrights \grants \nu
\activechans'. Q^B)$, where both are \wa processes such that $P^A \labrel P^B$,
Furthermore assume that for all $ \pubof{k} \in \refrights$
we have  $\attKnowledge^A \ngetPower_\refrights \secof{k}$.
Then the processes are annotated equivalent,  $P^A \approx_l P^B$ (with the annotations removed).
\end{restatable}

\subsection{Main theorems}
\label{sec:results}

Our first security property grants confidentiality for each created variable in the following way.
When a variable $x$ is created with a protection $R$, this implicitly defines a set of devices which are allowed
to access $x$. 
If the untyped attacker is not in this set, 
then any collaboration of the attacker with the denied-access devices can not 
learn any information about the value stored by the variable:
they cannot detect an arbitrary modification of the variable.

\begin{restatable}{theorem}{thmA}\label{prop:att-never-sees-external-change}
 Let $P= A  ~|~ D_1 ~|~ \ldots ~|~ D_n $ be a valid initial process.
 We consider a reduction $P \ra\!^* P'$ with
 $P'=\nu \activechans. A' ~|~ D'_1 ~|~ \ldots ~|~ D'_n$ such that for some $1 \leq i \leq n$, 
 $D'_i$ is $\dev{\mem}{\paral{\newvar{x}{S_R}{E};C}{C'}}$.
 Let $\LR=\mem(R)$ be the set of keys corresponding to $R$, 
 let $D_H = \devicesFrom{\LR}{P\ra P'}$  
 and let $\attKnowledge$ be the unions of memories of the devices of $P'$ whose indexes are not in $D_H$ and of device $A'$.
 If $D_H$ is safe (as stated in \autoref{def:dev-from}),
 then $(\attKnowledge \grants \subp{P'}{D_H})\approx_l (\attKnowledge \grants \subp{P'_{i:x=v}}{D_H})$.
\end{restatable}

\begin{proof}
According to \autoref{prop:identifiers}, we have 
$\attKnowledge_0 \grants \subp{P}{D_{H}} \Lra{l_1,\ldots,l_{n}}^* \attKnowledge' \grants \subp{P'}{D_{H}}$
where $\attKnowledge_0 = \bigcup_{i \notin D_H}{\mem_i}$.
As $P$ is a valid initial process,
we get that $\subp{P}{D_{H}}$ is also a valid initial process 
(the verification of all conditions to be a valid initial process is immediate).
Since $D_H$ is safe, we consider $\refrights_0$, $\refrights$
and the \standard annotated semantics
from \autoref{prop:safe->standard}:
\[\attKnowledge_0; \refrights_0 \grants \subp{P}{D_{H}} \Lra{l_1,\ldots,l_{n}} 
\attKnowledge; \refrights \grants \subp{P'}{D_H}\]
where $\LR \subseteq \refrights$ and 
 \begin{equation}\label{eq:1}
 \forall \pubof{k} \in \refrights_0,~ \attKnowledge_0 \ngetPower_{\refrights_0} \secof{k}.
 \end{equation}
From \autoref{prop:init-well-annotated}, the annotation of the initial process is \wa.
From \autoref{lemma:subred} on this annotation, we get for all $i \in D_H$:
$\safeDev{\refrights}{D_i}$ 
and, due to \eqref{eq:1}, for all $\pubof{k} \in \refrights$ we have $\attKnowledge \ngetPower_\refrights \secof{k}$.
Since $\LR \subseteq \refrights$, and given the \eqref{type:assign} rule, we also have 
\[\safe{\refrights}{\high}{\mem \addm{\{\loch{x}{v'}}\}}{\assign{x}{v}}.\]
Therefore, we satisfy all conditions of \autoref{def:final-relation}:
\[(\attKnowledge; \refrights \models  P') \labrel (\attKnowledge ; \refrights \models P'_{i:x=v}).\]
Finally, we conclude with \autoref{prop:=>labelled}.
\end{proof}

\begin{example}
From the example of \autoref{sec:example}, 
let consider the variable 
\trlinebreak
$\mathit{sharedSecret} : \Int_{\{\mathit{serverKey}, \mathit{Alice}, \mathit{Bob}\}}$
of the device $SD$.
The \autoref{prop:att-never-sees-external-change} states that only devices
from $\devicesFrom{\{\mathit{serverKey}, \mathit{Alice}, \mathit{Bob}\}}{\mathit{System} \ra \mathit{System'}}$
can distinguish between $\mathit{sharedSecret}=42$ or $\mathit{sharedSecret}=43$.
Whatever are the reductions, these devices contains for sure
$Srv$ since its key is known from the beginning and $\mathit{SD}$ since it creates $\mathit{Alice}$.
The only threat is that the channel $\mathit{otherPrin}$ has not been established with $\mathit{MD}$
or that $c$ has not been established with the device $\mathit{RD}$ of Bob.
These threats could be removed by an additional authentication protocol 
or the use of a private channel (for instance if the connection between $\mathit{SD}$ and $\mathit{MD}$ is a direct wired connection).
This means that not knowing which code is run on $\mathit{Srv}$ and $\mathit{RD}$ is not a threat
as long as $\mathit{Srv}$ and $\mathit{RD}$ guarantee that their codes are well-typed.
\end{example}

Finally, we state a standard security property in the simpler case where there is no untyped attacker.
Given a process $P$, if, at some point, we change the value that a
variable is bound to in memory, where that variable has been typed
with right $R$,
then new reductions will not alter values in memory locations that
have been typed strictly less confidential than $R$ (i.e., those
variables have rights that contain public keys not contained in
$R$).

\begin{restatable}{theorem}{thmB}\label{thm:taint-marking}
 Let $P$ be a valid initial process.
 We consider a reduction $P \ra P'$ with
 $P'=D'_1 | \dots | D'_n$ such that for some $i$, 
 $D'_i$ is $\dev{\mem}{\paral{\newvar{x}{S_R}{E};C}{C'}}$.
 Let $v$ be any value of type $S$.
 Let $P''=D''_1| \dots|D''_n$ such that $P' \ra^* P''$ 
 then there exists $Q''$ 
 such that $P'_{x=v} \ra^* Q''$ 
 where for all memory location $y$ that have been created in any device with rights $R_l$ with
 $R \nlpc R_l$, we have $\erase{R}{v_P}=\erase{R}{v'_Q}$ where $v_P$ is the value of $y$ in $P''$
 and $v_Q$ the value of $y$ in $Q''$.
\end{restatable}

{\bf Implementation:} We have implemented an interpreter for this language in Ocaml\footnote{https://github.com/gazeaui/secure-type-for-cloud}.  
The program strictly follows the extended semantics, 
it has commands to add a new device; to do an internal reduction on the selected thread(s) of the selected device(s);
to perform an attacker communication,
and to type-check each device of the system according to the annotations. 
To define a device which starts with keys in memory,
two additional commands are provided which are allowed only as a preamble:
$\mathbf{load~principal} ~p~ \mathbf{ from }~ i;$
and $\mathbf{load } ~x~ :~ \pubKeyType~ \mathbf{ from }~ i;$
where $i$ is a number; an identical $i$ on two devices represents a shared key. This implementation demonstrates that the syntax is well-defined and effective, 
and allows us to test the invariance of the properties with demonstrative examples.
The example of \autoref{sec:example} has been tested 
using this implementation:
we are able to reduce the process such that Bob gets the secret from Alice
and we can verify that each step correctly type-checks.

\section{Conclusion}
\label{sec:concl}

We have presented a security type system that provides location based correctness results 
as well as a more traditional non-interference result. 
The key novelty of our system is to allow principal identities to be created and shared across devices in a controlled way, without the need for a global PKI or middleware layer.
Hence, our correctness result states that well-typed devices can control which other devices may acquire their data, even in the presents of untyped attackers. 
We have illustrated our system with an example of an open cloud server that accepts new users. 
This server does perform some monitoring of its users but our type system proves that it does not monitor the content of their data. 
We argued that our framework is particularly appropriate to cloud systems 
where organizations will want guarantees about where their data will be stored as well as the secrecy of their data.

\subsubsection*{Acknowledgement} We would like to thank Alley Stoughton for her help with this work; her insightful comments and useful advice greatly improved this paper.

\bibliographystyle{plain}
\bibliography{refs}

\newpage
\appendix

\section{Additional Language Rules}
\label{app:other-rules}

\begin{figure}
\begin{gather}
\inferrule{ v\!:\!T \in \Gamma}{ \typeVar{v} : T} \tag{var\_T}\label{type:var}
\\
\inferrule{ 
}{ \typeVar{i} : \Int_\bot} \tag{int\_T}\label{type:int}
\\
\inferrule{ P \in \GammaPrin}{ \typeVar{\publicKey{P}} : \pubKeyType_{\bot}} \tag{pubKey\_T}\label{type:pubKey}
\\
\inferrule{ \GammaPrin;\GammaKey;\Gamma \vdash e_1\!:\!\Int_{R_1} 
\\ \GammaPrin;\GammaKey;\Gamma \vdash e_2\!:\!\Int_{R_2}
}{\GammaPrin;\GammaKey;\Gamma \vdash e_1 + e_2 : \Int_{R_1 \cap R_2}} 
\tag{Sum\_T}\label{type:sum}
\\
\inferrule{ \typeVar{v_1\!:\!S_R} \dots  \typeVar{v_n\!:\!S_R}}{
\typeVar{ \{v_1, \dots v_n \} : \arrayType{S}_R}}
\tag{array\_T} \label{type:array} 
\\
\inferrule{ \typeVar{x\!:\! \arrayType{S}_{R_1}} \\ \typeVar{e \!:\!Int_{R_2}} 
}{
\typeVar{ x[e]\!:\!S_{R_1 \cap R_2}}}
\tag{element\_T} \label{type:element} 
\\
\inferrule{
}{
\GammaPrin; \GammaKey \vdash \bot
\tag{public\_T}
}
\\
%
\inferrule{
\pc;\GammaPrin;\GammaKey;\GammaChan; \Gamma \vdash C_1 \\
\pc;\GammaPrin;\GammaKey;\GammaChan; \Gamma \vdash C_2
}{
\pc;\GammaPrin;\GammaKey;\GammaChan; \Gamma \vdash ~C_1~|~C_2 }  
\tag{para\_T} \label{type:para}
\\
\inferrule{
}{
\pc;\GammaPrin;\GammaKey;\GammaChan; \Gamma \vdash \nop }  
\tag{skip\_T} \label{type:skip}
\\
\inferrule{\pc;\GammaPrin;\GammaKey;\GammaChan; \Gamma \vdash C
}{
\pc;\GammaPrin;\GammaKey;\GammaChan; \Gamma \vdash ~\bang C }  
\tag{bang\_T} \label{type:bang}
\\
\inferrule{
     \typeVar{e_1:\Int_{R_1}} \\
     \typeVar{e_2 : S_{R_2}} \\
     \typeVar{x : \arrayType{S}_{R_3}} \\
   \vprin \vdash R_3 \subseteq \pc \cap R_1 \cap R_2 \\
  \pc;\GammaPrin;\GammaKey;\GammaChan; \Gamma \vdash C 
}{
\pc;\GammaPrin;\GammaKey;\GammaChan; \Gamma \vdash  \assign{x[e_1]}{e_2}; ~C } \tag{assign\_array\_T}\label{type:assignArray}
\\
\inferrule{
 \pc = \bot \\
 \pc;\GammaPrin ;\GammaKey;\GammaChan \cup \{ c \!:\!\!\chanType{S_\bot}{\bot}\}; \Gamma  \vdash C 
}{
 \pc;\GammaPrin ;\GammaKey;\GammaChan; \Gamma \vdash \acceptPub{c}{\chanType{S_\bot}{\bot}}; ~C }   \tag{accept\_1\_T} \label{type:acceptPub}
\\
\inferrule{
p \in \GammaPrin \\
k \in \GammaKey \\
\rightsComp{  \{\publicKey{p},k\} }{ R_1 \subseteq R_2 \subseteq \pc}  \\
R_2; \GammaPrin; \GammaKey; \GammaChan \cup \{ c : \chanType{S_{R_1}}{R_2} \};  \Gamma \vdash C 
}{
\pc;\GammaPrin; \GammaKey; \GammaChan;  \Gamma \vdash \acceptCCert{c}{\chanType{S_{R_1}}{R_2}}{k}{p};~C }  
\tag{accept\_2\_T} \label{type:open-cert-server}
\end{gather}
\caption{Other type rules for commands}
\label{fig:other-type-rules-a}
\end{figure}

\begin{figure}
\begin{gather}
\inferrule{
\pc;\GammaPrin;\GammaKey;\GammaChan; \Gamma \vdash G \{C / \nop \}   
}{
\pc;\GammaPrin;\GammaKey;\GammaChan; \Gamma \vdash \sync\{G\};C }  
\tag{sync\_T} \label{type:sync}
\end{gather}
Where 
\[
\begin{array}{rcl}
G & ::= & \cond{e_1}{e_2}{G_1}{G_2} \\
& | & \new{x}{S_R}{e};~G  \\
& | & \assign{x}{e};~G \\
& | & \assign{x[e_1]}{e_2};~G\\
& | & \decrypt{P}{e}{x}{S_R}{G_1}{G_2} \\
& | & \register{P_1}{e}{P_2}{G_1}{G_2} \\
& | & \paral{G_1}{G_2} \\
& | & \nop 
\end{array}
\]
\caption{Other type rules for commands (continued)}
\label{fig:other-type-rules-b}
\end{figure}

\begin{figure}

\begin{gather}
 \inferrule{
 M(e) = i \\ i \in \{0,\ldots,n \}\\
 M(x) = \{v_0,\ldots,v_i,\ldots,v_n\}
 }{
 M(x[e]) = v_i
} \tag{element\_E}\label{red:arrayelement}
\\
\inferrule{
  M(e) = i \\ i \notin \{0,\ldots,n \} \\
  M(x) = \{v_0,\ldots,v_i,\ldots,v_n\}
 }{
  M(x[e]) = \NotAValue
 } \tag{element\_err\_E}\label{red:array_err}
\\
 \inferrule{
 M(P) =  \prin{\pubKey}{\secKey}{R} 
 }{ 
 M(\publicKey{P}) = \pubKey
} \tag{pub\_E}\label{red:pub}
\\
\inferrule{ 
   \dev{M}{C_1}  \rightarrow^*  \dev{M'}{\nop}
 }{
   \dev{M}{ \sync\{C_1\};C_2 }  \rightarrow  \dev{M'}{C_2}
}
 \tag{sync\_S}  \label{red:sync}
%
%
%
\\
%
 \inferrule{
 M(e_1) = v \\
 M(e_2) = v 
 }{
 \dev{M}{\cond{e_1}{e_2}{C_1}{C_2}}
 \rightarrow \dev{M}{C_1} 
 }\tag{if1\_S} \label{red:if1}
 \\
 \inferrule{
 M(e_1) = v \\
 M(e_2) = v' \\
  v \neq v' 
  }{
 \dev{M}{\cond{e_1}{e_2}{C_1}{C_2}}
 \rightarrow \dev{M}{C_2} 
 }\tag{if2\_S} \label{red:if2}
\\
%
 \inferrule{
 }{
 \dev{M}{~ \bang{C} }
 \rightarrow \dev{M}{ C~|~\bang{C}}}\tag{bang\_S} \label{red:bang}
\\
%
%
 \inferrule{
M(e_1) = i \\
M(e_2) = v
 }{
 \dev{M\cup\{x \mapsto \{v_1,\dots v_i, \dots v_n\}}{\assign{x[e_1]}{e_2}; ~C}
 \\ \rightarrow \dev{M \cup \{x \mapsto  \{v_1,\dots v, \dots v_n\}\}}{C}
 }\tag{assignArray\_S} \label{red:assignArray}
\\
 \inferrule{
M(p) = \prin{k^+}{k^-}{\LR_2} \\
M(e) = \enc{v}{n}{\LR} \\
k^+ \not\in M(RS_1) \lor M(RS_1) \not\subseteq \LR
 }{ 
 \dev{M}{\decrypt{p}{e}{x}{S_{RS_1}}{C_1}{C_2} }
 \rightarrow 
 \dev{M }{C_2}
} \tag{dec\_false\_S}\label{red:decF}
\end{gather}
\caption{Other semantics rules}
\label{fig:other-semantics-rules}
\end{figure}


Reduction is defined modulo the standard equivalence rules, which may be applied in any context, e.g.\\

$D_1 ~|~ D_2 \equiv D_2 ~|~ D_1$ \\
$D_1 ~|~ (D_2 ~|~ D_2) \equiv (D_1 ~|~ D_2) ~|~ D_2$ \\
$\dev{M_1}{C_1~|~C_2} \equiv \dev{M_1}{C_2~|~C_1}$ \\
$\dev{M_1}{C_1~|~(C_2~|~C_3)} \equiv \dev{M_1}{(C_1~|~C_2)~|~C_2}$ \\

The rules not mentioned in \autoref{fig:semantics-rules} can be found in  \autoref{fig:other-semantics-rules}
and the types rules not mentioned in \autoref{fig:types-rules-cmd} and
\autoref{fig:types-rules-expr} can be found in
\autoref{fig:other-type-rules-a} and \autoref{fig:other-type-rules-a}.

%
%

%

\begin{figure}
\[
\begin{array}{rcll}
\attExpr & ::= & m_1 & \text{a variable of $\attKnowledge$}
\\ & | & \mathrm{encr}_{\attExpr_2}(\attExpr_1)  & \text{encrypt $\attExpr_1$ with rights $\attExpr_2$}
\\ & | & \mathrm{decr}_{\attExpr_1}(\attExpr_2) & \text{decrypt $\attExpr_1$ with the secret key $\attExpr_2$}
\\ & | & \mathrm{alter}(\attExpr_1,\attExpr_2) & \text{the principal $\attExpr_1$ where its rights are now $\attExpr_2$}
\\ & | & \mathrm{freshprin}(\attExpr) & \text{the $\attExpr$-th principal generated by the attacker}
\\ & | & \mathrm{pub}(\attExpr) & \text{the public key of a principal $\attExpr$}
\\ & | & \mathrm{sec}(\attExpr) & \text{the secret key of a principal $\attExpr$}
\\ & | & \mathrm{rights}(\attExpr) & \text{the rights of a principal $\attExpr$}
\\ & | &  \{\attExpr_1, \dots \attExpr_n \} & \text{tuple}
\\ & | &  \attExpr_1[\attExpr_2] & \text{element of a tuple}
\\ & | & \attExpr_1 \oplus \attExpr_2  & \text{where $\oplus$ is $+,-,\times,\dots$}
\\ & | & i & \text{an integer}
\end{array}
\]

\begin{gather}
 \inferrule{
  M(\attExpr_1) = v
 \\ M(\attExpr_2) = \{ \pubKeyOf{k_1} ,\ldots ,\pubKeyOf{k_n}\}
 \\ \fresh(n)
 }{
  M(\mathrm{encr}_{\attExpr_2}(\attExpr_1)) = \enc{v}{n}{\{ \pubKeyOf{k_1} ,\ldots ,\pubKeyOf{k_n}\}}
 } \tag{attacker\_encryption}\label{eexpr:encr}
 \\
 \inferrule{
  M(\attExpr_1) = \enc{v}{n}{\{ \pubKeyOf{k_1} ,\ldots ,\pubKeyOf{k_n}\}}
 \\ M(\attExpr_2) = \secKeyOf{k_i} \quad i \in \{1,\ldots,n\}
 }{
  M(\mathrm{decr}_{\attExpr_2}(\attExpr_1)) = v
 } \tag{attacker\_decryption}\label{eexpr:decr}
 \\
\inferrule{
  M(\attExpr_1) = \prinv{k}{r}
 \\ M(\attExpr_2) = r'
 }{
  M(\mathrm{alter}({\attExpr_1},\attExpr_2)) = \prinv{k}{r}
 } \tag{attacker\_alter}\label{eexpr:alter} 
 \\
\inferrule{
  M(\attExpr) = i \in \mathbb{N}
 }{
  M(\mathrm{freshprin}({\attExpr})) = \prinv{k_i}{\{\}}
 } \tag{attacker\_fresh}\label{eexpr:fresh} 
\end{gather}

\caption{extended expressions of the attacker and some of their reduction rules.}\label{fig:extended-expressions}
\end{figure}

\section{Proofs of correctness}
\label{app:proofs-correctness}
Here we provide the proof of all propositions of \autoref{sec:result}.

\subsection{From rights to allowed devices}
No proposition in this section.

\subsection{Definition of the attacker and of the open process semantics}

\subprocess*

\begin{proof}
By induction on the number of reduction steps.
We assume that the result is true for $n-1$ reductions
i.e we have $P ~|~ A \ra\!^* P' ~|~ A'$
which is mimicked by 
\[\attKnowledge \grants \subp{P}{S} \lra{l_1,\ldots,l_m} \subp{\attKnowledge' \grants P'_a}{S}\]
and we prove that if $P' \mid A' \ra P'' \mid A''$ then we can also have
\[\attKnowledge \grants \subp{P'_a}{S} \lra{l?} \subp{\attKnowledge' \grants P''_a}{S}.\]

\begin{enumerate}
 \item 
If the last reduction involves devices which are part of $S$, then this same reduction can be done on $\subp{P'_a}{S}$.
\item If the last reduction involves devices which are not part of $S$ then 
no reduction has to be done in the open semantics.
%
\item If the last reduction $\nu \activechans.c. P \mid A \ra P_1 \mid A_1$ is \eqref{red:i/o} 
from outside $S$ to $S$ transferring $v$ on $c$ then
the labeled reduction 
$\attKnowledge \grants \nu \activechans.P \{\ca{c} / c \} 
\lra{\LabIn{\ca{c},v}} \attKnowledge \grants \nu \activechans.P_1 \{\ca{c} / c \}$ is possible
since $v$ has been produced from values into memory not in $S$ and the attacker can derive all these values.
\item If the last reduction is a channel establishment between a device from $S$ and a one not from $S$,
the proof is similar. Note that the new channel name is annotated with $\ca{c}$ instead of $c$.
\item If the last reduction $\nu \activechans.c. P \mid A \ra P_1 \mid A_1$ is \eqref{red:i/o} 
from $S$ to outside $S$ transferring $v$ on $c$ then
the labeled reduction $\attKnowledge \grants \nu \activechans.P \lra{\LabOut{c,v}} \attKnowledgeAdd{v} \grants \nu \activechans.P_1$
exists and we still have for all $v'$ in memory of a device of $P$ which identifier is not in $S$,
$ \attKnowledgeAdd{v} \vdash v'$.
\end{enumerate}
\end{proof}

\subsection{Extended syntax with extra annotations}

\safeStandard*

\begin{proof}
 We assume that $P'$ is the $n$-th reduction of $P$.
 We denote by $P_m$, $P$ after $m$ reductions.
 Let $(I_i, \LR_i)= \backward{P \ra P_1 } (I_{i+1},\LR_{i+1})$
 for $0 \leq i < n$ and  $(I_n, \LR_n)=(\emptyset, \LR)$.
   
 We now prove by induction on the number of reduction steps
 that the annotation 
 $\attKnowledge_0; \refrights_0 \grants \subp{P}{D_H} 
 \Lra{l_1,\ldots,l_m} \attKnowledge_m; \refrights_m \grants \subp{P'_m}{D_H}$
 such that $\refrights_m = \bigcup_{0\leq k \leq m} \LR_i$ is \standard.
 
 Note that as the set is safe, 
 we have $\forall \pubof{k} \in \refrights_0,~ \attKnowledge_0 \ngetPower_{\refrights_0} \secof{k}$.
 Moreover we also $\LR \subseteq \refrights_n$ by definition.
 
 $\LR_m \subseteq \refrights_m$ where $\LR_m$ contains all the keys of $\LR$ which have already been created.
 
 The only reduction we need to consider are the one which might not be \standard i.e.
 when a reduction which creates a new principal $\pvalue=\prinv{k}{\{k_1, \ldots, k_n\}}$
 where $k$ is in $\bigcup_{0\leq k \leq m} \LR_i$.
 However, by design of the $\backward{}$ function, we have $\{k_1, \ldots, k_n\} \in \refrights_{m-1}$:
 the reduction is \standard.
\end{proof}

%

In the following auxiliary lemma, 
we state that \lowr expressions reduces to \lowr values while \highr expressions reduces
to \highr values.
\begin{lemma}\label{lemma:tag-on-expression}
 Let $e$ and $\mem$ be such that the sub-expressions $e_i : R_i$ of $e$ are such 
 that $\evaluates{\mem}{e_i}{\tagx{v_i}{\ell_i}}$  with $x_i=\high$ iff. $\refrights \lpc R_i$.
 Let $\tagx{v}{\ell}$ such that $\evaluates{\mem}{e}{\tagx{v}{\ell}}$.
 We have: $e : S_R$ with $\refrights \lpc R$ iff $\ell=\high$.
\end{lemma}

\begin{proof}
 If $\refrights \lpc R$, we consider all constructors:
 \begin{enumerate}
 \item $e=a+b$. 
 If $e : S_R$ with $\refrights \lpc R$, due to the typing rule 
 $a$ or $b$ has type $\Int_{R'}$ with $\refrights \lpc R'$,
 by the induction hypothesis, one of them reduced to a \highr value and 
 the reduction rule states that the result has to be \highr.
 If $e : S_R$ with $\refrights \nlpc R$ due to the typing rule
 both $a$ and $b$ has type $\refrights \nlpc R'$,
 by hypothesis both of them reduce to a \lowr value, 
 the reduction rule states that the result will not be \lowr.
 \item $e=x[e']$. 
  If $e : S_R$ with $\refrights \lpc R$, due to the typing rule \ref{type:element}
 $x$ or $e'$ has type with label $\refrights \lpc R'$,
 by the induction hypothesis, one of them reduced to a \highr value and 
 the reduction rule states that the result has to be \highr.
 If $e : S_R$ with $\refrights \nlpc R$, due to the typing rule
 both $x$ and $e'$ has type $\refrights \nlpc R'$,
 by hypothesis both of them reduce to a \lowr value, 
 the reduction rule states that the result will be \lowr.
 \item $e=\encE{e_1}{R}$.
  According to \ref{type:enc}, we have $e : S_\bot$ and reduces to a \lowv value (since the new root is \lowv).
 \item $e= \release{P}$, according to \ref{type:release}, $e :\SecK_\bot$ and reduces to a \lowv value.
 \item $e=\publicKey{P}$, according to \ref{type:pubKey}, $e :\PubK_\bot$ and reduces a \lowv value.
\end{enumerate}
\end{proof}

\begin{lemma}\label{lemma:no-tag<=>low-type}
 Let $\refrights,C,\mem,y$ such that $\safe{\refrights}{\mem}{y}{C}$ for some $\ctx$ and
 where $C$ is a command with an expression $e$ such that $\typeVar{e} : S_R$.
 We have that $\refrights \lpc R$ iff. $\evaluates{\mem}{e}{\tagged{v}}$ for some value $v$.
\end{lemma}

\begin{proof}
By induction on the structure of $e$.
If $e$ is a variable, we conclude from Item \ref{safe:variables} of \autoref{def:safe}.
For the inductive case, we use \autoref{lemma:tag-on-expression} to conclude.
\end{proof}

%
%

\subjectreduction*

\begin{proof}
W.l.o.g. we consider that the reduction is done on $D_1$ (eventually $D_1$ and $D_2$ if it is a communication between two device).
We write $D_1=\dev{\mem_1}{\thr{x}{C_1} ~|~ \ldots ~|~ \thr{x_n}{C_n}}$ and
$D'_1=\dev{\mem'_1}{\thr{x'}{C'_1} ~|~ \ldots ~|~ \thr{x_n}{C'_n}}$.
Let ${\GammaPrin}_i$, ${\GammaKey}_i$, $\ctx_i$ and $\pc_i$ such that $\safe{\refrights}{x_i}{\mem_i}{C_i}$.
We prove all the conditions of \wa separately.
\begin{enumerate}
\item[\ref{safe:well-type}.] {\bf $C_i'$ is well-typed for some $\pc_i',\ctx_i'$}

 By hypothesis, we know that $C_i$ is well-typed for some ${\GammaPrin}_i$, ${\GammaKey}_i$, $\pc_i,\ctx_i$.
 If $C_i := C  \mid  C'$ then the two new threads after reduction are well-typed with ${\GammaPrin}_i$, ${\GammaKey}_i$, $\ctx_i$ and $\pc_i$
 due to \eqref{type:para}.
 If $C_i := !C$ then the two new threads are also well-typed due to \eqref{type:bang}.
 Otherwise, the form of $C_i$ is $c ; C'_i$ (up to name renaming) 
 or $\textrm{cond}(c)\THEN C_a \ELSE C_b$ where $C_a$ or $C_b$ is
 an renaming of $C'_i$.
 All the corresponding type rules enforces that $C_i$ is well-typed 
 only if $C'_i$ is well-typed.
 
\item[\ref{safe:variables}.] {\bf  For all location $x$ in $\mem$, we have 
  \begin{itemize}
   \item either $\locl{x}{\untagged{v}}$ in $\mem$ for some value $v$ and $x \in \GammaPrin \cup \GammaKey$
  \item or $\GammaPrin;\GammaKey;\ctx \vdash x : S_R$ and
  \begin{itemize}
   \item if $\refrights \lc \mem(R)$, $\loch{x}{\tagged{v}}$ in $\mem$, 
   \item if $\refrights \nlc \mem(R)$, $\locl{x}{\untagged{v}}$ in $\mem$ for some value $v$. 
  \end{itemize}
\end{itemize}}

First, we consider the annotation of variables.
Once a memory location is created, 
there is no reduction rule to change its type ($\lowr$ or $\highr$),
$\refrights$ never becomes more confidential ($\refrights$ increases only with new fresh keys) 
and all types rules enforce that the reduced command $C'$ can be typed 
with a context which preserves rights for all variables.
For these reasons, if the property holds when a variable is created, it holds forever.
During creation, by design of the annotated rules, the location is created as \highr or \lowr 
only depending on if $\refrights \lpc R$ for regular variable and is always \lowr
when the location contains a principal or a key.
For inputs, since the input command does not specify the type, 
we look at the annotation of the channel,
which is itself defined from the condition $\refrights \lpc R_1$ in the annotated semantic rules 
about channel creation
\eqref{red:h:open-for}, \eqref{red:lh:open-for}, \eqref{red:ll:open-for}
for channel opening with the attacker, 
from the fact that $\attKnowledge \ngetPower_\refrights \secof{k}$ for all $\pubof{k} \in \refrights$ 
we know that the channel has a type for values less confidential than $\refrights$.

For the annotation of values, we need to check that no \highv values are stored into \lowr memory
(the annotation always produces \highv values from \highr memory).
 The reductions which set new values to \lowr memory locations are the following: 
 \eqref{red:l:aff}, \eqref{red:l:new_var}, \eqref{red:let},
 \eqref{red:ll:dec1}, \eqref{red:l:i/o} and \eqref{redatt:input}.
 For the first rules (\eqref{red:l:aff}, \eqref{red:l:new_var} ), 
 the stored value is the evaluation of an expression $E$.
 Since the type rules associated to these commands enforces $E$ to be of a type $R$ such that $\refrights \lpc R$, 
 according to Lemma \ref{lemma:no-tag<=>low-type}, $E$ reduces to a \lowv term.
 The argument is the same for \eqref{red:let}.

For the internal communication \eqref{red:l:i/o}, 
the argument is the same except that the type judgment of $E$ is on the other device.
For the input from the attacker \ref{redatt:input}, 
the rule set that attacker's inputs are \lowv.

For rule \eqref{red:ll:dec1} where one layer of encryption is removed, 
we have to ensure that the decrypted value is \lowv.
The typing rule requires that $\refrights \nlpc R_1$ and $R_2 \lc R_1$
which implies $\refrights \nlpc R_2$.
From Item~\ref{safe:cipher-variable}, we conclude.

The rule \eqref{red:hl:dec1} can never be applied. 
Indeed, if the cipher is \highv, from \autoref{lemma:no-tag<=>low-type}, the cipher has a \highr type,
then the typing rule enforce the \progc to be high which means the thread has to be \hight.

\item[\ref{safe:pc}.] {\bf $\refrights \lc \pc$ if and only if $y=\high$.}

First note that all types rules of commands requires that their continuation can be typed with the same \progc.
The interesting cases are when the type rules requires the continuation to be typed with a more restrictive \progc
and, for the other direction of the equivalence, when the reduction rules reduces a \lowt thread into a \hight thread.

\begin{itemize}
 \item Conditional is typed with \eqref{type:if}: the \progc $\pc'$ of the continuation if set to $\refrights \lc \pc'$
 when at least one of the expressions is more confidential than $\refrights$ (we assume that initially $\refrights \nlpc \pc$).
 According to \autoref{lemma:no-tag<=>low-type}, one of the expression reduces to a \highv value.
 In this case, the reductions rules \eqref{red:h:if1} and \eqref{red:h:if0} sets the thread to \hight.
 Reciprocally, if the thread is set to \hight, on of the value is \highv and according to \autoref{lemma:no-tag<=>low-type}
 the type of the continuation is $\refrights \lc \pc'$.
 \item Secure channel opening is typed with \eqref{type:open-cert-client} and \eqref{type:open-cert-server}.
 Both type rules requires the continuation to have a \progc $\pc'$ set such that $\refrights \lc \pc'$ iff
 $\refrights \lc R_2$ where $R_2$ is the right to see the channel.
 The reduction rule \eqref{red:h:open-for} can be applied only when $R_2$ and $R'_2$, the rights to see the channel of both threads
 are equal and sets both threads to \hight exactly with the same condition.
 \item Decryption is typed with \eqref{type:dec}, the \progc $\pc'$ of the continuation if set to $\refrights \lc \pc'$
 when the expression to decrypt has rights $R$ with $\refrights \lc R$.
 This coincides with the reductions rules \eqref{red:hh:dec1}, \eqref{red:hl:dec1},\eqref{red:h:dec0}
 where the thread becomes \hight iff the decrypted value is \highv (due to \autoref{lemma:no-tag<=>low-type}). 
\end{itemize}

\item[\ref{safe:channel}.] {\bf For all $c$ in the thread such that $\ctx \vdash c : \Chan({S}_{R_1})_{R_2}$, 
  the annotation of $c$ is $c^z_{(z')}$ 
  where $z=\high$ iff. $\refrights \lc R_2$, $z'=\high$ iff. $\refrights \lc R_1$.}

  The annotation on channel is definitively fixed once it is created.
  The channel establishment rules \eqref{red:h:open-for}, \eqref{red:lh:open-for},
  \eqref{red:lh:open-for} and \eqref{red:ll:open-for} all enforce this property in their premisses.
  Public channels have type $\bot$ and are annotated with $c^-_\low$.

\item[\ref{safe:no-high-attacker-chan}] {\bf For all $\ca{c}$, 
we have $\ctx \vdash \ca{c} : \Chan({S}_{R_1})_{R_2}$ 
  with $\refrights \nlc R_2$ and $\refrights \nlc R_1$.}

The interesting case is the creation of a new channel with the attacker.
This can be done through rules \eqref{redatt:auth-server-chan-for} 
and \eqref{redatt:auth-client-chan}.
By hypothesis, the attacker does not know any key in $\refrights$ so it can only provides a secret key $\secof{k}$
not in $\refrights$.
Due to \eqref{type:open-cert-client} or \eqref{type:open-cert-server},
$R_1$ and $R_2$ should contains $\pubof{k}$.
Therefore neither $R_1$ or $R_2$ is a subset of $\refrights$, which means $\refrights \nlc R_2$ and $\refrights \lc R_1$
as $\vprin$ is required to be empty. 

\item[\ref{safe:low-pc-low-chan}] {\bf If $y$ is $\low$, then there is no $\chanhh{c}$ in $C$.}

This is true before reduction. Since establishing a new channel $\chanhh{c}$ set the thread
to \hight, the property is also true after reduction.

\item[\ref{safe:cipher-variable}]{ \bf For all $x$ in memory, if a sub-term of $x$ is $\enc{\tagged{t}}{n}{RS}$
  then $\refrights \lc RS$.}
  
The encryption typing rule for $E=\encE{E_1}{R}$, \eqref{type:dec}, requires that $E_1 :R_1$ is such that $R_1 \lc R$.
From \autoref{lemma:no-tag<=>low-type}, 
we have $\evaluates{\mem}{E_1}{\tagged{v}}$ implies $\refrights \lc R_1$ 
and therefore $\refrights \lc R$ which allows to conclude.

\item[\ref{safe:released-prin}] {\bf For all  $x$ in memory, if a sub-term of $x$ matches $\enc{\prin{k}{\LR}}{n}{\LR'}$
  then $\LR = \LR'$ or $\pubof{k} \notin \refrights$.}
  
  Encapsulated principals can only be produced with a $\release{P}$ expression or by an input of the attacker.
  In the first key, the release reduction \eqref{rede:release} always encrypts the principal $\prinv{k}{\LR}$
with the keys of $\LR$ in the other case since by hypothesis we have $\attKnowledge \ngetPower_{\refrights} \secof{k}$
we also have $\attKnowledge \ngetPower_{\refrights} \prin{k}{r}$ for any $r$ so the attacker cannot output this value.
\end{enumerate}

Finally, we prove that when the annotation is \standard, we have $\attKnowledge' \ngetPower_{\refrights'} \secof{k}$
for all $\pubof{k} \in \refrights'$. 
The result is immediate for reductions which do not change $\attKnowledge$ or $\refrights$.

The set $\attKnowledge$ increases only with the reduction \eqref{redatt:output}.
The values which contains secret keys of $\refrights$ are of type $\privKeyType$ 
(or a combination of tuples and encryptions of this type).
According to item \ref{safe:released-prin}, any  $\prinv{k}{R}$ is encrypted with $R$.
Moreover, by definition of \standard annotation, all keys $\prinv{k_i}{R_i}$ in $\refrights$ 
are such that all keys of $R_i$ are also in $\refrights$.
Since $\attKnowledge \ngetPower_\refrights k$, the attacker cannot decrypt the key.

The reduction which modify $\refrights$ is \eqref{red:h:nwp}, but the new key is fresh
so the attacker cannot derive new facts from the fact that $\attKnowledge \ngetPower_\refrights' v$
is syntactically more powerful than  $\attKnowledge \ngetPower_\refrights v$.
\end{proof}


%

 \initialprocess*

\begin{proof}
Condition \autoref{safe:low-pc-low-chan} of \wa is true since $\pc=\bot$ and the tread is $\low$. 
Since $\mem_i$ and $\ctx$ only contains principals and public keys, the other conditions which concerns 
other kind of values and variables are immediately true. 
\end{proof}

\subsection{Labelled bisimilarity}

\begin{lemma}\label{lemma:memrel=>eval-equal}
 When two memories $\mem_1$ and $\mem_2$ only contains \lowr locations,
 $\mem_1 \nvdash \secof{k}$ and  $\mem_2 \nvdash \secof{k}$ for all $k$ such that $\pubof{k} \in \refrights$
and  $\mem_1 \memrel{\refrights} \mem_2$ 
 implies that, for all extended attacker expression $f$ (\autoref{fig:extended-expressions}),
 $\erase{\refrights}{\mem_1(f)}=  \erase{\refrights}{\mem_2(f)}$.
\end{lemma}

\begin{proof}
 By induction on the structure of $f$. 
 This is true if $f$ is a variable or an integer.
 We assume that $\erase{\refrights}{\mem_1(f_i)}=\erase{\refrights}{\mem_2(f_i)}$ for all sub-expressions $f_i$ of $f$.
 For $\mathrm{freshprin}(i)$, $i$ has to evaluate has an integer. Since
 the function $\mathrm{freshprin}(i)$ returns the $i$-th principal of a determined sequence of principals,
 the expression returns the same key in both memory.
 For $\mathrm{encr}_r(f)$, since the extended expression $r$ has to evaluate to a right it evaluate to the same right for both
 memories, so the result is either obfuscated with both memories or none of them.
 For $\mathrm{decr}_s(f)$, the result is immediate if $f$ evaluates to a cipher which is not obfuscated.
 Otherwise, $f$ is a cipher encrypted with rights $r'$, $\refrights \lc r'$, by hypothesis, 
 we have that $s$ cannot corresponds to any key in $r'$ therefore the decryption fails (result $\NotAValue$) in both memories.
\end{proof}

\begin{lemma}\label{prop:memrel=>static-equiv}
 When two memories $\mem_1$ and $\mem_2$ only contain \lowr locations, 
 $\mem_1 \nvdash \secof{k}$ and $\mem_2 \nvdash \secof{k}$ for all $k$ such that $\pubof{k} \in \refrights$
and  $\mem_1 \memrel{\refrights} \mem_2$ 
 implies that $\mem_1 \approx_s \mem_2$.
\end{lemma}

\begin{proof}
 Let $f$ and $g$ be two extended expressions which evaluates to the same value $v$ under $\mem_1$ ($\mem_1(f)=\mem_1(g)=v$). 
 We show that $\mem_2(f)=\mem_2(g)$.
 From \autoref{lemma:memrel=>eval-equal}, we have 
 $\erase{\refrights}{\mem_2(f)}=\erase{\refrights}{\mem_2(g)}=\erase{\refrights}{v}$.
 However the obfuscation function does not obfuscate the nonce of ciphers.
 And the attacker power does not allow the attacker to generate a cipher with a nonce which is not fresh
 (this is due to the fact that nonce cannot be extracted from ciphers from which the key is unknown).
 Therefore each nonce is uniquely associated to a unique value.
 As nonces are equals the obfuscated values are equals too.
%
\end{proof}

 \begin{lemma}\label{lemma:no-tag=>equiv}
 Given a command context $C$, a set of public keys $\refrights$, 
 two memories $\mem$ and $\mem'$, an expression $E$ and $\ell' \in \{ \low, \high \}$,
if we have $\safe{\refrights}{\mem}{\ell'}{C[E]}$ and $\safe{\refrights}{\mem'}{\ell''}{C[E]}$, $\mem \memrel{\refrights} \mem'$
 , $\evaluates{\mem}{E}{\untagged{v}}$ and $\evaluates{\mem'}{E}{\tagx{v'}{\ell}}$ 
 then $\erase{\refrights}{v}=\erase{\refrights}{v'}$.
\end{lemma}

\begin{proof}
By induction on the structure of $E$.
\begin{enumerate}
 \item 
If $E$ is a variable $x$ with $\evaluates{\mem}{E}{\untagged{v}}$ then due to Item \ref{safe:variables} of Definition \ref{def:safe}, 
we have $\locl{x}{\untagged{v}}$.
By definition of $\mem \memrel{\refrights} \mem'$, we conclude $\evaluates{\mem'}{E}{\untagged{v}}$.
\item
If $E$ is $e_1 + e_2$ with  $\evaluates{\mem}{E}{\untagged{v}}$, due to rule \eqref{rede:l:sum} and \eqref{rede:h:sum}
$v$ is \lowv implies that none of its sub-expressions is \highv, we conclude by induction.
\item For $E=x[n]$, the principle is the same. 
\item For $E=x\{E_1,\ldots,E_n\}$, the principle is also the same. 
\item If $E = \encE{e}{\LR'}$ with $\evaluates{\mem}{E}{v}$. 
If $\refrights \nlpc \LR'$, due to the type rule \eqref{type:enc}, 
$e : B_R''$ with $\mem(R'') \lpc \LR' \lpc \refrights$. 
According to \autoref{lemma:no-tag<=>low-type} $e$ is \lowv, 
we conclude by induction.
If $\refrights \lpc R'$, then $\erase{\refrights}{\enc{v}{n}{\LR'}}=\enc{\_}{n}{\LR'}$ so $v$ does not appear.
\item For expressions which involves principals ($\mathbf{pub}(p)$ and $\mathbf{release}(p)$) we conclude as principals are not stored in \highr location.
\end{enumerate}

\end{proof}

%

\begin{restatable}{proposition}{stronger}\label{thr:stronger}
Let $\attKnowledge_1; \refrights \grants P_1$ and  $\attKnowledge_2; \refrights \grants P_2$ 
be two \wa processes such that
$\forall \pubof{k} \in \refrights,\quad \attKnowledge_i \ngetPower_\refrights \secof{k} $ ($i \in \{1,2\}$) and 
$P_1 \labrel P_2$, we have:
 \begin{enumerate}
  \item if $P_1 \lra{l} P_1'$ 
  then $P_2 \lra{l} P_2'$ and $P_1' \labrel P_2'$ for some $P_2'$. 
  \item if $P_1 \ra P_1'$, then $P_2 \ra^? P_2'$ and $P_1' \labrel P_2'$ for some $P_2'$,
  \end{enumerate}
\end{restatable}


\begin{proof}
\newcommand{\PA}{P^A}
\newcommand{\PB}{P^B}
To avoid conflict of indexes, we consider $\PA$ and $\PB$ instead of $P_1$ and $P_2$.
\paragraph{When the reduction is a labeled reduction with the attacker}

\begin{itemize}
 \item Case where the reduction rule applied to $\PA$ is \eqref{redatt:pub-chan} 
 on a thread $\thr{x}{\connectPub{\ca{c^A}}{\Chan(S_\bot)_\bot};{C'^A}}$.
 From \eqref{type:connectPub}, we know that $\refrights \lc \pc$.
 From Item \ref{safe:pc} of \autoref{def:safe}, we know that this reduction is possible only when
 $x = \low$.
 Hence there exists the same thread in $\PB$. 
 The reduction on the corresponding thread of $\PB$ leads to $C'^A$ through the same label $\LabIn{c,S}$ 
 up to the name of $\ca{c^B}$,
 but they are equal up to alpha-renaming of the channel name.
Hence we have the condition on commands of \autoref{def:final-relation}.
 Finally the memory is unchanged and the attacker does not learn any value any both process:
 the condition about memories of \autoref{def:final-relation} and memory equivalence of the attackers hold.
 
 
 \item Case where the reduction rule applied to $\PA$ is \eqref{redatt:auth-server-chan-for},
 on a thread $\thr{x}{\connectCCert{\ca{c^A}}{\Chan(S_{R_1})_{R_2}}{k}{P};{C'^A}}$,
 with label $\LabIn{\ca{c^A},\secof{k_s},\pubof{k_c},S,r,r'}$ be the values in the label of the reduction.
 Since $\pubof{k_s} \notin \refrights$ (according to our hypothesis), 
 due to \eqref{redatt:auth-server-chan-for}, $k \notin \refrights$ and
 due to the type rule \eqref{type:open-cert-server}, we have 
 $\refrights \nlpc R_1$ and $\refrights \nlpc R_2$.
 According to Item \ref{safe:pc} of \autoref{def:safe}, the thread is \lowt ($x=\low$).
 So an identical thread exists in $\PB$, and since the label enforces $S$, $R_1$ and $R_2$ to be the sames,
 it reduces to $C'^A$ up to the renaming of the channel name.
 Finally, $\attKnowledge^A$, $\attKnowledge^B$, $\mem^A$ and $\mem^B$ remain the sames. 
 
 \item The case where the reduction rule applied to $\PA$ is \eqref{redatt:auth-client-chan} is similar.
 
  \item Case where the reduction rule applied to $\PA$ is \eqref{redatt:output}
  on a thread $\thr{z}{\outputChan{\ca{c}}{e};C^A}$, with label $\LabOut{\ca{c^A},x}$.
 According to Item \ref{safe:no-high-attacker-chan} of \autoref{def:final-relation}
 the channel is of type $\Chan({S}_{R_1})_{R_2}$ with $\refrights \nlpc R_1$ and $\refrights \nlpc R_2$.
 According to \eqref{type:output}, we have $R_2 \nlpc \pc$ and so $\refrights \nlpc \pc$.
 According to Item \ref{safe:pc} of \autoref{def:safe} the thread is \lowt ($z=\low$).
 Therefore there exists an identical thread in $\PB$ where the same label reduction is possible.
 Moreover, according to \eqref{type:output}, we have $e :R_1$ with $\refrights \nlc R_1$ so 
 according to \autoref{lemma:no-tag<=>low-type}, $\evaluates{\mem^A}{e}{\untagged{v^A}}$ for some $v^A$.
 According to Lemma \ref{lemma:no-tag=>equiv},
 we have $\erase{\refrights}{v^B}=\erase{\refrights}{v^A}$, hence $\attKnowledge^A \memrel \attKnowledge^B$.
 
 \item Case where the reduction rule applied to $\PA$ is \eqref{redatt:input}
 on a thread $\thr{z}{\inputChanII{\ca{c}}{x};C^A}$ with label $\LabIn{\ca{c^A},f}$.
 According to Item~\ref{safe:no-high-attacker-chan} of \autoref{def:final-relation}
  the channel is of type $\Chan({S}_{R_1})_{R_2}$ with $\refrights \nlpc R_1$ and $\refrights \nlpc R_2$.
  According to \eqref{type:input}, we have $R_2 \nlpc \pc$ and so $\refrights \nlpc \pc$.
  According to Item \ref{safe:pc} of \autoref{def:safe} the thread is \lowt ($z=\low$).
  So there is an identical thread in $\PB$ (which reduces with the same label).
  According to \autoref{lemma:memrel=>eval-equal} we have
  $\erase{\refrights}{\attKnowledge^A(f)}=\erase{\refrights}{\attKnowledge^A(f)}$.
  Therefore, the new memories $\mem'^A$ and $\mem'^B$ are equivalent.  
\end{itemize}

\paragraph{When the reduction is made on a thread where $\safe{\refrights}{\low}{\mem}{C}$}

In that case, we use the fact 
for any thread $\thr{\low}{C^A}$ on some device with memory $\mem^A$ on which the reduction happens, 
there exists a thread $\thr{\low}{C^B}$ on the other process with memory $\mem^B$ 
such that $C^A \simeq C^B$ and $\mem^A \memrel{\refrights} \mem^B$.
 
We prove all the requirements sequentially: first, we prove that $\refrights$ reduces to the same $\refrights'$
in both process.
Secondly, we prove the commands reduce to equivalent commands.
Thirdly, we prove the equivalence of the new memories.

\begin{enumerate}
 \item {\bf We have ${\refrights}'^A = {\refrights}'^B$}
 The one reduction that changes $\refrights$ is \eqref{red:h:nwp}.
 When reference rights change in a reduction due to \eqref{red:h:nwp}, 
 the thread is \lowt due to the typing rule so the reduction can also occurs in process $B$.
 We conclude by renaming the new key variables of both processes to the same name. 

\item {\bf Reductions rules for $C^B$ is the same as for $C^A$: $C'^A=C'^B$  (or both new threads are \hight)}

There are two cases for ``if'', and two cases for ``dec''
\begin{itemize}
	\item The rule  $C^A \ra C'^A$ is \eqref{red:h:if1} or \eqref{red:h:if0}. 
	In these cases, according to Lemma \ref{lemma:no-tag<=>low-type},
	the evaluated expressions was \highr.
	So it is also true for $C^B$ therefore 
	the conditional expression of $C^B$ also evaluates 
	to a \highv value according to Lemma \ref{lemma:no-tag<=>low-type}.
	On the other hand, the new thread $C'^A$ is $\high$
	so whatever $C^B$ reduces to the then or the else branch, 
	the condition about commands of \autoref{def:final-relation} is satisfied.
	\item The rule $C^A \ra C'^A$ is \eqref{red:l:if1} or \eqref{red:l:if0}.
	This means the expression evaluates to a \lowv value.
	In addition, the evaluation returns an integer which implies the expression
	can only contains $+$, integers and variables.
	This implies that the results is only dependent on \lowv values:
	since \lowv values in $C^A$ and $C^B$ are the same both reduce according to the same rule.
	\item The rule is \eqref{red:l:dec0}, \eqref{red:ll:dec1} or \eqref{red:lh:dec1}.
	Since the value to decrypt is \lowv, the value in $A$ and $B$ differ only on \highv sub-terms
	which means that the encryption key is the same, so both processes take the same branch.
	\item The rule is \eqref{red:h:dec0} or \eqref{red:hh:dec1}: 
	whatever \eqref{red:l:dec0} or \eqref{red:ll:dec1} is taken the thread become \hight on $A$ and $B$.
\end{itemize}
\item {\bf New values in memories are such that $\mem'^A \memrel{\refrights} \mem'^B$.}
First of all, because the rights are the same in $C^A$ and $C^B$
and premisses on the rules depends only on rights, the high or low rule is applied for both command.
	\begin{itemize}
	\item Rules that store values into \highr memory locations
          \eqref{red:h:new_var}, \eqref{red:h:aff}, 
          \trlinebreak
	\eqref{red:lh:dec1}, \eqref{red:hh:dec1}:
	they affect \highr locations so condition $\mem'^A \memrel{\refrights} \mem'^B$ holds whatever is the stored value.
	\item Rules that store values into \lowr memory locations: 
	\eqref{red:l:new_var}, \eqref{red:l:aff}, \eqref{red:ll:dec1}, \eqref{red:let}
	according to Item \ref{safe:variables} of Definition \ref{def:safe}, these values are not \highv
	so, according to Lemma \ref{lemma:no-tag=>equiv}, we have
	$\erase{\refrights}{v^A}=\erase{\refrights}{v^B}$.
	\item Rule \eqref{red:l:register1}: $K^A$ is in a \lowr memory location so $K^A=K^B$ 
	and so the new principals are identical.  
	\item Creates a principal \eqref{red:h:nwp} and \eqref{red:l:nwp}: 
	the new value contains fresh values and rights that are identical in $A$ and $B$.
 	\end{itemize}
\end{enumerate}

\paragraph{When the reduction is done on a thread $C$ where $\safe{\refrights}{\high}{\mem}{C}$}

\begin{enumerate}
\item Internal reductions does not change the equivalence class of the process 
(so the other process can simulate it with no reduction).
Indeed, first, there is no reduction that reduces a \hight thread to a \lowt thread: 
the reduced command $C'$ does not impact the equivalence class of the process.
Secondly, the reductions that create or change memories location are: 
\eqref{red:h:new_var}, \eqref{red:hh:dec1} and \eqref{red:h:aff}.
In each case, the type rules states that $\pc \lpc R$ where $R$ is the right of the memory.
Due to Item \ref{safe:pc} of Definition \ref{def:safe}, we have
that $\pc$ is high. 
Hence $R$ is high too, so the rule either creates or change a memory location with flag $\high$.
\end{enumerate}


\paragraph{When the reduction is made between two threads $C_1$ and $C_2$ where $C_1$ is \hight.}

\begin{itemize}
 \item If both $C_1$ and $C_2$ are \hight. The only possible reduction are \eqref{red:h:i/o} and \eqref{red:h:open-for}.
In case of \eqref{red:h:open-for}, from Item \ref{safe:low-pc-low-chan}, we have $\refrights \lc \pc$.
The types rules \eqref{type:open-cert-client} and \eqref{type:open-cert-server}
enforces that $R_1$ and $R_2$ are \highr. Due to Item \ref{safe:channel}, the channel is annotated with $\chanhh{c}$.
In case of \eqref{red:h:i/o}, the type rules \eqref{type:input} enforces that all received values are stored
into \highv memory. Since \hight threads and \highr memory locations are free to differ between the two processes,
the equivalence is preserved after reduction.

\item If $C_2$ is \lowt. Any annotation variant of the rule \eqref{red:i/o} 
is not possible when one thread is \lowt and the other is \hight.
Indeed, according to Item \ref{safe:channel} of Definition \ref{def:safe}, 
there is no $c_\high^\high$ channel on the \lowt thread.
On the other \hight thread, a $c_\low^\low$ or $c_\low^\high$ channel has a \lowr type according to 
Item \ref{safe:channel} of Definition \ref{def:safe},
so according to \eqref{type:output} and \eqref{type:input} they cannot be used one this thread.

So only establishing connection is possible with rule \eqref{red:h:open-for} 
where the client is on the \lowt thread $C_2$ 
and the server is on the \hight thread $C_1$.
Due to the reduction rule \eqref{red:h:open-for}, $C_2'$ is a \hight thread
and $C'_1$ is \hight too so the equivalence is preserved in both processes. 
\end{itemize}

\paragraph{When the reduction is made on two threads which are \lowt.} 
\

For all communication opening rules, rights in both process are identical since the public keys of the rights set are public.
So both processes thread reduces either to a \hight thread or a \lowt thread.
Hence the two reduced processes remains in the same equivalent class.

For the input/output reduction rule \eqref{red:l:i/o}, according to Item~\ref{safe:low-pc-low-chan} of Definition \ref{def:safe},
the channel type can only be $c^\low_\low$ or $c^\high_\low$ since the threads are \lowt. 
On the output thread, \highv values are sent only on $c^\high_\low$ channels.
According to Item~\ref{safe:channel}, on the input side $c^\high_\low$ has type $\Chan(S_R)_{R'}$ with $\refrights \lc R$.
According to \eqref{type:input}, the value is stored in a \highv memory location so
even if the values differ between $P^A$ and $P^B$ memories are still equivalent.
On the other hand, if the values $v^A$ and $v^B$ sent on the channels are not \highv then they have to be equivalence:   
$\erase{\refrights}{v^A}=\erase{\refrights}{v^B}$. 
\end{proof}

\implylabelled*

\begin{proof}
 The first item is a direct consequence of \autoref{prop:memrel=>static-equiv}
 The two last items are a consequence of \autoref{thr:stronger} with the fact that
  when the annotation is \standard, according to \autoref{lemma:subred},
 we have  that $\forall \pubof{k} \in \refrights,\quad \attKnowledge_i \ngetPower_\refrights \secof{k} $ 
 is preserved by the reduction, so the proposition can be generalized to any number of reduction steps.
\end{proof}


\subsection{Main theorems}

\thmB*

\begin{proof}
Since there is no untyped attacker, we can use non-\standard annotations.
We consider the one such that at the beginning $\refrights$ contains all the keys of $R$ which are present in some memory,
and which add a key in $\refrights$ if and only if this key is in $R$.
From \autoref{thr:stronger} on this annotated process,
we get the existence of $Q$ such that $Q \labequiv P''$ which imply that 
for all memory $\mem_i$ in $P''$ and $\mem'_i$ in $Q$, we have $\mem_i \memrel{R} \mem'_i$.
By definition of $\memrel{R}$, we conclude.
\end{proof}

\section{Annotated rules}\label{sec:ext-rules}
\small
 We define $\topc(x,y)$ to be $\low$ except if $x=y=\high$.

\begin{gather}
\inferrule{\fresh{(k)} \\ \fresh(P')
 }{
 \attKnowledge;  \refrights \grants \devA{\mem}{\thr{x}{\newPrin{P}{R} ; ~C}}
\\ \ra \attKnowledge; \refrights \cup \{ \pubof{k} \}
\grants \devA{\mem \addm{ \loc{P'}{\prinv{k}{R}}}}{\thr{x}{C\{P'/P\}}}
 } \tag{high newPrin}\label{red:h:nwp}
 \\
  \inferrule{\fresh{(k)} \\ \fresh{P'}
 }{
 \attKnowledge;  \refrights \grants \devA{\mem}{\thr{x}{\newPrin{P}{R} ; ~C}}
\\ \ra \attKnowledge; \refrights \grants \devA{\mem \addm{ \loc{P'}{\prinv{k}{R}}}}
{\thr{x}{C\{P'/P\}}}
 } \tag{low newPrin}\label{red:l:nwp}  
 \\
   \inferrule{
\evaluates{\mem}{e}{\tagx{\enc{\prinv{k_1}{R_1}}{R}{n}}{a}} \\
\evaluates{\mem}{P_2}{\prinv{k_2}{R_2}} \\
\pubof{k_2} \in R_1 \\
\fresh(P_3)
 }{
 \devA{\mem}{\thr{x}{\register{P_2}{e}{P_1}{C_1}{C_2}}} \\
 \ra \devA{\mem \addm{\locl{P_3}{\prinv{k_1}{R_1}}}}{\thr{x}{C_1\{P_3 / P_1\}}}
 }\tag{low register}\label{red:l:register1}
  \end{gather}
 
  \begin{gather}
 \inferrule{
 \evaluates{\mem}{p}{\prinv{k}{r}} \\
 \evaluates{\mem}{e}{\tagx{\enc{v}{n}{RS}}{a}} \\
  RS \lpc \mem(R_1)
  \\ \pubof{k} \in \mem(R_1) \\
  \fresh(y)
 \\ \refrights \nlpc \mem(R_1) 
 }{ 
 \attKnowledge; \refrights \grants \devA{\mem}{\thr{z}{\decrypt{p}{e}{x}{S_{R_1}}{C_1}{C_2}}}
 \\ \ra \attKnowledge; \refrights \grants \devA{\mem \addm{\locl{y}{v}}}{\thr{z}{C_1 \{y /x \}}}
 }\tag{low dec\_true}\label{red:ll:dec1}
 \\
 \inferrule{
 \evaluates{\mem}{P}{\prinv{k}{r}} \\
 \evaluates{\mem}{e}{\tagged{\enc{v}{n}{RS}}} \\
  RS \lpc \mem(R_1)
  \\ \pubof{k} \in \mem(R_1) \\
  \fresh(y)
 \\ \refrights \nlpc \mem(R_1)
 }{
 \attKnowledge; \refrights \grants \devA{\mem}{\thr{z}{\decrypt{P}{e}{x}{S_{R_1}}{C_1}{C_2}}}
\\ \ra \attKnowledge; \refrights \grants \devA{\mem\addm{\locl{y}{v}}}{\thr{\high}{C_1 \{y /x \}}}
 }\tag{high thread/low var dec\_true}\label{red:hl:dec1}
 \\
\inferrule{
 \evaluates{\mem}{P}{\prinv{k}{r}} \\
 \evaluates{\mem}{e}{\enc{v}{n}{RS}} \\
  RS \lpc \mem(R_1)
  \\ \pubof{k} \in \mem(R_1) \\
  \fresh(y)
 \\ \refrights \lc \mem(R_1)
 }{ 
 \attKnowledge; \refrights \grants \devA{\mem}{\thr{z}{\decrypt{P}{e}{x}{S_{R_1}}{C_1}{C_2}}}
 \\ \ra \attKnowledge; \refrights \grants \devA{\mem\addm{\loch{y}{v}}}{\thr{z}{C_1\{y /x \}}}
 }\tag{low thread/high var dec\_true}\label{red:lh:dec1}
 \\
 \inferrule{
\evaluates{\mem}{P}{\prinv{k}{r}} \\
 \evaluates{\mem}{e}{\tagged{\enc{v}{n}{RS}}} \\
  RS \lpc \mem(R_1)
  \\ \pubof{k} \in \mem(R_1) \\
  \fresh(y)
 \\ \refrights \lpc \mem(R_1)
 }{ 
 \attKnowledge; \refrights \grants \devA{\mem}{\thr{z}{\decrypt{P}{e}{x}{S_{R_1}}{C_1}{C_2}}}
\\ \ra \attKnowledge; \refrights \grants \devA{\mem\addm{\loch{y}{v}}}{\thr{\high}{C_1\{y /x \}}}
 }\tag{high dec\_true}\label{red:hh:dec1}
 \\
 \inferrule{
\evaluates{\mem}{P}{\prinv{k}{r}} \\
 \evaluates{\mem}{e}{{\enc{v}{n}{RS}}} \\
 (RS \nlpc \mem(R_1) \vee \pubof{k} \notin \mem(R_1))
 }{ 
 \attKnowledge; \refrights \grants \devA{\mem}{\thr{z}{\decrypt{P}{e}{x}{S_{R_1}}{C_1}{C_2}}}
 \\ \ra \attKnowledge; \refrights \grants \devA{\mem}{\thr{z}{C_2}}
 }\tag{low dec\_false}\label{red:l:dec0}
 \\
 \inferrule{
\evaluates{\mem}{P}{\prinv{k}{r}} \\
 \evaluates{\mem}{e}{\tagged{\enc{v}{n}{RS}}} \\
  (RS \nlpc \mem(R_1) \vee \pubof{k} \notin \mem(R_1))
 }{ 
 \attKnowledge; \refrights \grants \devA{\mem}{\thr{z}{\decrypt{P}{e}{x}{S_{R_1}}{C_1}{C_2}}}
\\ \ra \attKnowledge; \refrights \grants \devA{\mem}{\thr{\high}{C_2}}
 }\tag{high dec\_false}\label{red:h:dec0}
 \end{gather}

%
  \begin{gather}
 \inferrule{
 \evaluates{\mem}{e}{\tagx{v}{t}} \\
 \fresh(y)
 \\ \refrights \subseteq \mem(R) 
 }{
 \attKnowledge; \refrights \grants \devA{\mem}{\thr{z}{\newvar{x}{S_R}{e} ; ~C}}
 \\ \ra \attKnowledge; \refrights \grants \devA{\mem \addm{\loch{x}{\tagged{v}}}}{\thr{z}{C \{y/x\}}} 
 } \tag{high new}\label{red:h:new_var}
 \\
 \inferrule{
 \evaluates{\mem}{e}{\tagx{v}{t}}
 \\ \refrights \varsubsetneq \mem(R)
 \\\fresh(y)
 }{
 \attKnowledge; \refrights \grants \devA{\mem}{\thr{z}{\newvar{x}{S_R}{e} ; ~C}}
 \\ \ra \attKnowledge; \refrights \grants \devA{\mem \addm{\locl{x}{\tagx{v}{t}}}}{\thr{z}{C}} 
 } \tag{low new} \label{red:l:new_var}
 \\
 \inferrule{
 \evaluates{(\mem\addm{\locl{x}{\tagx{v}{t}}})}{e}{\tagx{v'}{t'}}
 }{
 \attKnowledge; \refrights \grants \devA{\mem \addm{\locl{x}{v}}}{\thr{z}{\assign{x}{e} ; ~C}}
 \\ \ra \attKnowledge; \refrights \grants \devA{\mem, \locl{x}{\tagx{v'}{t'}}}{\thr{z}{C}} 
 } \tag{low assign}\label{red:l:aff}
 \\
 \inferrule{
 \evaluates{(\mem\addm{\loch{x}{\tagged{v}}})}{e}{\tagx{v'}{t}}
 }{
 \attKnowledge; \refrights \grants \devA{\mem \addm{\loch{x}{\tagged{v}}}}{\thr{z}{\assign{x}{e} ; ~C}}
 \\ \ra \attKnowledge; \refrights \grants \devA{\mem \addm{ \loch{x}{\tagged{v'}}}}{\thr{z}{C}} 
 }\tag{high assign} \label{red:h:aff}
 \\
 \inferrule{
 \evaluates{\mem}{e_1}{\untagged{v}} 
 \\ \evaluates{\mem}{e_2}{\untagged{v}} 
}{
 \attKnowledge; \refrights \grants \devA{\mem}{\thr{z}{\cond{e_1}{e_2}{C_1}{C_2}}}
 \\ \ra \attKnowledge; \refrights \grants \devA{\mem}{\thr{z}{C_1}}
 } \tag{low if\_true}\label{red:l:if1}
 \\
 \inferrule{
 \evaluates{\mem}{e_1}{\untagged{v}} 
 \\ \evaluates{\mem}{e_2}{\untagged{v'}} 
 \\ v \neq v'
 }{
 \attKnowledge; \refrights \grants \devA{\mem}{\thr{z}{\cond{e_1}{e_2}{C_1}{C_2}}}
 \\ \ra \attKnowledge; \refrights \grants \devA{\mem}{\thr{z}{C_2}}
 } \tag{low if\_false}\label{red:l:if0}
 \\
 \inferrule{
 \evaluates{\mem}{e_1}{\tagx{v}{t}}
 \\ \evaluates{\mem}{e_2}{\tagx{v'}{t'}}
 \\ \topc(t,t')=\high
 }{
 \attKnowledge; \refrights \grants \devA{\mem}{\thr{z}{\cond{e_1}{e_2}{C_1}{C_2}}}
 \\ \ra \attKnowledge; \refrights \grants \devA{\mem}{\thr{\high}{C_1}}
 }\tag{high if\_true} \label{red:h:if1}
 \\
 \inferrule{
\evaluates{\mem}{e_1}{\tagx{v}{t}}
 \\ \evaluates{\mem}{e_2}{\tagx{v'}{t'}}
 \\ v \neq v' \\ \topc(t,t')=\low
 }{
 \attKnowledge; \refrights \grants \devA{\mem}{\thr{z}{\cond{e_1}{e_2}{C_1}{C_2}}}
 \\ \ra \attKnowledge; \refrights \grants \devA{\mem}{\thr{\high}{C_2}}
 }\tag{high if\_false} \label{red:h:if0}
 \end{gather}

 \begin{gather}
\inferrule{ 
\fresh(c)
\\ \evaluates{\mem_1}{P_s}{\prinv{K_s}{R_s}}
\\ \evaluates{\mem_2}{P_c}{\prinv{K_c}{R_c}}
\\ \mem_1({R_1}) =  \mem_2({R'_1})
\\ \mem_1(R_2) = \mem_2(R'_2)
\\ \evaluates{\mem_1}{k_c}{\pubof{K_c}}
\\ \evaluates{\mem_2}{k_s}{\pubof{K_s}}
\\ \refrights \lpc \mem_1(R_1) 
\\ \refrights \lpc \mem_1(R_2) 
  }{
   \attKnowledge; \refrights \grants \devA{\mem_1}{\thr{z}{\acceptCCert{c_1}{\Chan({S}_{R_1})_{R_2}}{k_c}{P_s}; ~C_1 }}
   \\ | \devB{\mem_2}{\thr{z'}{\connectCCert{c_2}{\Chan({S}_{R'_1})_{R'_2}}{k_s}{P_c} ; ~C_2}}
 \\ \ra 
   \attKnowledge; \refrights \grants \devA{\mem_1}{\thr{\high}{C_1 \{\chanhh{c}/c_1\}}}
 \\ | \devB{\mem_2}{\thr{\high}{C_2 \{\chanhh{c}/c_2\}}}
 } \tag{high open\_priv} \label{red:h:open-for}
 \\
  \inferrule{ \fresh(c)
\\ \evaluates{\mem_1}{P_s}{\prinv{K_s}{R_s}}
\\ \evaluates{\mem_2}{P_c}{\prinv{K_c}{R_c}}
\\ \mem_1({R_1}) =  \mem_2({R'_1})
\\ \mem_1(R_2) = \mem_2(R'_2)
\\ \evaluates{\mem_1}{k_c}{\pubof{K_c}}
\\ \evaluates{\mem_2}{k_s}{\pubof{K_s}}
\\ \refrights \nlpc \mem_1(R_1) 
\\ \refrights \nlpc \mem_1(R_2) 
}{
   \attKnowledge; \refrights \grants \devA{\mem_1}{\thr{z}{\acceptCCert{c_1}{\Chan({S}_{R_1})_{R_2}}{k_c}{P_s}; ~C_1}}
   \\ | \devB{\mem_2}{\thr{z'}{\connectCCert{c_2}{\Chan({S}_{R'_1})_{R'_2}}{k_s}{P_c} ; ~C_2}}
 \\ \ra 
   \attKnowledge; \refrights \grants \devA{\mem_1}{\thr{z}{C_1 \{\chanll{c}/c_1\}}}
\\ | \devB{\mem_2}{\thr{z'}{C_2 \{\chanll{c}/c_2\}}}
 } \tag{low open\_priv} \label{red:ll:open-for}
 \\
\inferrule{ 
  \fresh(c)
\\ \evaluates{\mem_1}{P_s}{\prinv{K_s}{R_s}}
\\ \evaluates{\mem_2}{P_c}{\prinv{K_c}{R_c}}
\\ \mem_1({R_1}) =  \mem_2({R'_1})
\\ \mem_1(R_2) = \mem_2(R'_2)
\\ \evaluates{\mem_1}{k_c}{\pubof{K_c}}
\\ \evaluates{\mem_2}{k_s}{\pubof{K_s}}
\\ \refrights \nlpc \mem_1(R_1) 
\\ \refrights \lpc \mem_1(R_2) 
 }{
   \attKnowledge; \refrights \grants \devA{\mem_1}{\thr{z}{\acceptCCert{c_1}{\Chan({S}_{R_1})_{R_2}}{k_c}{P_s}; ~C_1}}
   \\ | \devB{\mem_2}{\thr{z'}{\connectCCert{c_2}{\Chan({S}_{R'_1})_{R'_2}}{k_s}{P_c} ; ~C_2}}
 \\ \ra 
   \attKnowledge; \refrights \grants \devA{\mem_1}{\thr{z}{C_1 \{\chanlh{c}/c_1\}}}
\\ | \devB{\mem_2}{\thr{z'}{C_2 \{\chanlh{c}/c_2\}}}
 } \tag{low thread/high value open\_priv} \label{red:lh:open-for}
\\
 \inferrule{
 \evaluates{\mem_1}{e}{v} \\
 \fresh(y)
 }{
 \attKnowledge; \refrights \grants \devA{\mem_1}{\thr{z}{\outputChan{\chanxh{c}}{E}; ~C_1}}
\\ | \devB{\mem_2}{\thr{z'}{\inputChan{\chanxh{c}}{x}; ~C_2}}
 \\ \ra
 \attKnowledge; \refrights \grants \devA{\mem_1}{\thr{z}{C_1}}
\\ | \devB{\mem_2 \addm{\loch{y}{\tagged{v}}}}{\thr{z'}{C_2 \{ y/x\}}}
 }\tag{high i/o} \label{red:h:i/o}
 \\
 \inferrule{
 \evaluates{\mem_1}{e}{v}\\
 \fresh(y)
 }{
 \attKnowledge; \refrights \grants \devA{\mem_1}{\thr{z}{\outputChan{\chanll{c}}{e}; C_1}}
\\ | \devB{\mem_2}{\thr{z'}{\inputChan{\chanll{c}}{X}; C_2}}
\\ \ra
 \attKnowledge; \refrights \grants \devA{\mem_1}{\thr{z}{C_1}}
\\ | \devB{\mem_2 \addm{\locl{y}{v}}}{\thr{z'}{C_2 \{y/x\}}}
 }\tag{low i/o} \label{red:l:i/o}
\end{gather}
%
%
\begin{gather}
   \inferrule{
   \evaluates{\mem}{P}{\prinv{k}{R}} \\ R \neq \{\} \\ \fresh(n)
 }{ 
\evaluates{\mem}{\release{P}}{\untagged{\enc{\prinv{k}{R}}{n}{R}}}
 }\tag{\anne{release}}\label{rede:a:release}
\\
\inferrule{
 \evaluates{\mem}{e}{\tagx{v}{t}}
 \\ \fresh(n)
 }{ 
\evaluates{\mem}{\encE{e}{RS}}{\untagged{\enc{\tagx{v}{t}}{n}{RS})}} 
 } \tag{\anne{enc}}\label{rede:a:enc}
\\
 \inferrule{
\evaluates{\mem}{e_1 }{\tagged{v_1}}
\\ \evaluates{\mem}{e_2 }{\tagx{v_2}{t}}
}{
 \evaluates{\mem}{e_1 + e_2}{\tagged{v_1+v_2}} 
 }\tag{\higha{sum}}\label{rede:h:sum}
 \\
 \inferrule{
\evaluates{\mem}{e_1 }{\untagged{v_1}}
\\ \evaluates{\mem}{e_2 }{\untagged{v_2}}
}{
 \evaluates{\mem}{e_1 + e_2}{\untagged{v_1+v_2}} 
 }\tag{\lowa{sum}}\label{rede:l:sum}
\end{gather}
%
%
\begin{gather}
\inferrule{
\fresh(c)
}{ \attKnowledge; \refrights \grants \devA{\mem_1}{\thr{x}{\connectPub{c_1}{\Chan(S_\bot)_\bot} ; ~C_1}}  \lra{\LabIn{c,S}}
\\   \attKnowledge; \refrights \grants \devA{\mem_1}{\thr{x}{C_1 \{\ca{c}/c_1\}}}
 }
 \tag{\atta{open\_public}}\label{redatt:pub-chan}
 \\
 \inferrule{ 
 \evaluates{\attKnowledge}{f}{\{ \secof{k_s},\pubof{k_c},S,r,r' \}}
 \\ \evaluates{\mem_2}{R}{r}
 \\ \evaluates{\mem_2}{R'}{r'} 
\evaluates{\mem_2}{k}{\pubof{k_s}} 
 \\ \evaluates{\mem_2}{P}{\prinv{k_c}{r'}} 
 \\ \fresh(c)
 }{
   \attKnowledge; \refrights \grants \devA{\mem_2}{\thr{x}{\connectCCert{c_2}{\Chan(S_{R_1})_{R_2}}{k}{P}; ~C_2}}
 \\ \lra{\LabIn{c,f} }
  \attKnowledge; \refrights \grants \devA{\mem_2}{\thr{x}{C_2 \{\ca{c}/c_2\}}}
 } \tag{\atta{open\_priv\_client}} \label{redatt:auth-client-chan}
\\
 \inferrule{ 
  \evaluates{\attKnowledge}{f}{\{ \pubof{k_s},\secof{k_c},S,r,r' \}}
 \\ \evaluates{\mem_1}{k}{\pubof{k_c}} 
 \\ \evaluates{\mem_1}{P}{\prinv{k_s}{r'}} 
 \\ \evaluates{\mem_1}{R}{r}
 \\ \evaluates{\mem_1}{R'}{r'} 
 \\ \fresh(c)
 }{
   \attKnowledge; \refrights \grants \devA{\mem_1}{\thr{x}{\acceptCCert{c_1}{\Chan(S_{R_1})_{R_2}}{k}{P}; ~C_1}}
 \\ \lra{\LabIn{c,f}}
   \attKnowledge; \refrights \grants \devA{\mem_1}{\thr{x}{C_1 \{\ca{c}/c_1\}}}
 } \tag{\atta{open\_server}} \label{redatt:auth-server-chan-for}
 \\
 \inferrule{
 \evaluates{\mem_1}{e}{\tagx{v}{t}}
  }{
 \attKnowledge; \refrights \grants \devA{\mem_1}{\thr{z}{\outputChan{\ca{c}}{e}; C_1}}
\\ \lra{\LabOut{c,m}} 
 \attKnowledgeAdd{(\locl{m}{\untagged{v}})}; \refrights \grants \devA{\mem_1}{\thr{z}{C_1}}
 } \tag{\atta{output}} \label{redatt:output}
 \\
 \inferrule{
 \evaluates{\attKnowledge}{\attExpr}{\untagged{v}}
 \fresh(y)
 }{
 \attKnowledge; \refrights \grants \devA{\mem_2}{\thr{z}{\inputChan{\ca{c}}{x}; ~C}}
\\ \lra{\LabIn{c,\attExpr}}
 \attKnowledge; \refrights \grants \devA{\mem_2\addm{\locl{y}{v}}}{\thr{z}{C \{ y / x\}}}
 } \tag{\atta{input}} \label{redatt:input}
 \end{gather}


\end{document}